\documentclass{article}

\PassOptionsToPackage{round}{natbib}
\usepackage[preprint]{neurips_2026}

\usepackage{framed}
 
\usepackage{amssymb,amsthm,amsmath}
\usepackage{mathtools,mathrsfs}
\usepackage[ruled,linesnumbered]{algorithm2e}

\usepackage{bm}
\usepackage{float}
\usepackage{subfig}

\usepackage{multirow}

\usepackage[utf8]{inputenc} 
\usepackage[T1]{fontenc}    
\usepackage[
    colorlinks=true,
    linkcolor=blue,
    citecolor=blue,
    urlcolor=blue,
    filecolor=blue,
    bookmarks=true,
    bookmarksopen=true,
    bookmarksnumbered=true,
    pdfstartview=FitH,
    breaklinks=true
]{hyperref}
\usepackage{url}            
\usepackage{booktabs}       
\usepackage{amsfonts}       
\usepackage{nicefrac}       
\usepackage{tcolorbox}
\usepackage{microtype}      
\usepackage{pgfplots}
\usepgfplotslibrary{fillbetween,groupplots}
\usetikzlibrary{calc}
\usepackage{xcolor}         

\definecolor{tcpblue}{HTML}{4A9BD3}

\SetCommentSty{tcpstyle}
\usepackage{adjustbox}
\definecolor{plotblue}{HTML}{365F9C}
\definecolor{plotred}{HTML}{CB4149}
\definecolor{plotteal}{HTML}{168F91}
\definecolor{plotpurple}{HTML}{7C4DA5}
\pgfplotsset{compat=1.18,evalue axis/.style={
    font=\fontsize{7}{8}\selectfont,
    xlabel style={font=\fontsize{7}{8}\selectfont},ylabel style={font=\fontsize{7}{8}\selectfont},
    title style={font=\fontsize{7}{8}\selectfont},tick label style={font=\fontsize{7}{8}\selectfont},
    grid=none,axis line style={line width=.4pt},tick style={line width=.4pt},
    tick align=outside,tick pos=left,scaled ticks=false,
    /pgf/number format/fixed,/pgf/number format/precision=3,
    legend style={font=\fontsize{7}{8}\selectfont,draw=none,fill=none,
        /tikz/every even column/.append style={column sep=2pt}},
    unbounded coords=jump,
}}

\definecolor{deepyellow}{RGB}{204,153,0}

\usepackage[capitalize,noabbrev]{cleveref}
\newtheorem{theorem}{Theorem}[section]
\newtheorem{lemma}[theorem]{Lemma}

\newtheorem{proposition}[theorem]{Proposition}

\newtheorem{definition}[theorem]{Definition}

\crefname{assumption}{assumption}{assumptions}
\Crefname{assumption}{Assumption}{Assumptions}

\title{Improved E-Value Thresholds with Applications to Multiple and Sequential Testing}

\author{%
{\bfseries
Jiayi Shen$^{1}$\thanks{Equal contribution.} \quad
Houzhen Liu$^{1}$\footnotemark[1] \quad
Qirun Zeng$^{2}$ \quad
Manhin Poon$^{2}$} \\[0.2em]
{\bfseries
Weiqi Yang$^{1}$ \quad
Qixin Zhang$^{3}$ \quad
Weiwei Zhuang$^{1}$} \\[0.7em]
{\small
$^{1}$University of Science and Technology of China, Hefei, China} \\[-0.05em]
{\small
$^{2}$City University of Hong Kong, Hong Kong SAR, China} \\[-0.05em]
{\small
$^{3}$Nanyang Technological University, Singapore} \\[0.4em]
{\small
\texttt{\{ustcsjy2021, liuhz0919, ywqyw123\}@mail.ustc.edu.cn}} \\[-0.05em]
{\small
\texttt{\{qirun.zeng,manhinpoon4-c\}@my.cityu.edu.hk}} \\[-0.05em]
{\small
\texttt{qixin.zhang@ntu.edu.sg}
\quad
\texttt{weizh@ustc.edu.cn}}
}

\begin{document}

\maketitle

\begin{abstract}
{E-values provide a flexible framework for statistical inference, but the universal threshold $1/\alpha$ can be conservative when the null distribution has additional structure. We study how structural restrictions limit the worst-case concentration underlying Markov's inequality, using a nondecreasing density model to constrain tail allocation and an $L$-Lipschitz density model to control local concentration. The nondecreasing model gives to the minimax rejection threshold, and the Lipschitz condition yields a sharper closed-form threshold with an $O(L^{-1/2})$ relative improvement over $1/\alpha$. We further show that this threshold is asymptotically sharp as $\alpha\to0$, in the sense that no uniformly valid threshold can improve on it by a fixed positive $L$-dependent amount. We then incorporate the Lipschitz calibration into multiple testing procedures with FDR control and use the same structural information to construct calibrated conditional e-values for anytime-valid sequential inference. Numerical experiments illustrate the resulting gains in discoveries and the trade-offs of the sequential procedures.}
\end{abstract}


\section{Introduction}
\label  {sec:introduction}
An e-value is a nonnegative random variable $E$ satisfying $\mathbb E_0[E]\leq1$ under the null hypothesis; large values indicate evidence against the null~\citep{vovk2021values}. This expectation-based criterion provides a framework for constructing and combining statistical evidence on a common scale. E-values arise in likelihood-ratio and Bayes-factor testing~\citep{shafer2011test} and have also been studied for composite hypothesis testing~\citep{zhang2024existence}. Their use extends from individual hypotheses to collections of decisions: e-BH controls the false discovery rate (FDR) under arbitrary dependence~\citep{wang2022false}, and e-value-based methods provide confidence regions and bounds on false discoveries~\citep{vovk2023confidence}. Applications include online multiple testing~\citep{xu2024online} and knockoff-based variable selection~\citep{ren2024derandomised}. When observations arrive sequentially, conditional e-values and e-processes provide evidence that remains valid under optional stopping and continuation~\citep{grunwald2024safe,ramdas2023game}, with nonnegative supermartingales supplying time-uniform error guarantees~\citep{howard2020time}. In these settings, the evidence can be used as data accumulate, without fixing the stopping time in advance. Across these uses, statistical decisions depend on how the evidence is calibrated against a rejection threshold.
A central issue in applications of e-values is that the universal rejection threshold $1/\alpha$ may be conservative for structured null distributions. This threshold follows from Markov's inequality: $\mathbb P(E\ge 1/\alpha) \le \alpha$. The bound is attained by placing probability $\alpha$ at $1/\alpha$ and the remaining mass at zero, a construction that uses only the mean constraint and allows probability to concentrate in the most favorable way. Additional assumptions on the null distribution can therefore rule out or weaken this worst-case behavior and permit smaller valid thresholds. Recent work~\citep{blier2024improved} shows that restrictions such as decreasing or unimodal null densities can indeed sharpen the standard Markov calibration.

This naturally motivates us to ask:
\begin{tcolorbox}[
    colback=blue!4,
    colframe=blue!45!black,
    boxrule=0.8pt,
    arc=1.5pt,
    left=5pt,
    right=5pt,
    top=5pt,
    bottom=5pt
]\emph{How much can additional structural information sharpen the rejection threshold}? 
\end{tcolorbox}To answer this question, we focus on two features underlying the Markov worst case: placing tail mass near the rejection threshold and concentrating probability mass in sharp local spikes. These two features suggest two complementary ways of restricting the null distribution. More specifically, a nondecreasing density constrains how tail mass can be distributed, whereas an $L$-Lipschitz density limits how sharply probability mass can concentrate locally. The resulting thresholds show how these two forms of structure lead to different degrees of improvement over the Markov calibration.

We further extend the Lipschitz calibration to multiple testing and sequential inference. In multiple testing, the improved threshold is incorporated into e-BH procedures with FDR control. In sequential inference, Lipschitz information is used to calibrate a new conditional e-value sequence whose product forms an anytime-valid e-process. Simulations study the resulting gains in multiple testing, while a retrospective study of S\&P~500 constituents examines the behavior of the sequential construction and its detection trade-offs.  To summarize, we make the following contributions:
\begin{itemize}
    \item We analyze a model on a bounded interval with an atom at zero and a nondecreasing density on the positive part, and derive its exact minimax rejection threshold.

    \item Under an $L$-Lipschitz constraint on the null density, we derive the sharper threshold $T_{\alpha}^{(L)} = \frac{1}{\alpha}\left(1 - \frac{\sqrt{2}(1 - \alpha)^{3/2}}{3\sqrt{L}}\right)$, yielding an explicit relative reduction of order $L^{-1/2}$ from the standard threshold $1/\alpha$. We further establish small-$\alpha$ asymptotic sharpness, showing that this threshold cannot be uniformly lowered by any fixed positive $L$-dependent additive margin.

    \item We incorporate the improved thresholds into e-BH procedures and establish FDR control under independence. An e-value boosting construction provides FDR control under arbitrary dependence.

    \item We extend the Lipschitz framework to anytime-valid sequential inference through new evidence updates, yielding an LC-process with time-uniform error control.

    \item Numerical experiments evaluate the discovery gains from Lipschitz calibration in multiple testing and its detection behavior in sequential inference, including controlled Gaussian experiments and a retrospective S\&P 500 study.
\end{itemize}

\subsection{Related work}

E-values have been developed as a general measure of statistical evidence, with close connections to likelihood ratios, Bayes factors, test martingales, and game-theoretic statistics~\citep{shafer2011test,vovk2021values,ramdas2023game}. Recent work has further studied the construction and power of e-values for composite hypotheses~\citep{zhang2024existence}, as well as broader frameworks for safe and anytime-valid inference based on e-values and e-statistics~\citep{grunwald2024safe,perez2024statistics}. These developments provide the theoretical background for the use of e-values in both fixed-sample and sequential inference.

In fixed-sample multiple testing, e-BH and related procedures establish FDR control under broad dependence conditions~\citep{wang2022false}, extending the classical Benjamini--Hochberg framework~\citep{benjamini1995controlling,benjamini2001control}. Related developments include stepwise, closed-testing, and post hoc procedures~\citep{holm1979simple,simes1986improved,shaffer1986modified,marcus1976closed,blanchard2008two,blanchard2020post}. Subsequent work has developed broader confidence and discovery procedures based on e-values~\citep{vovk2023confidence}, online multiple testing~\citep{xu2024online}, and knockoff-based variable selection~\citep{ren2024derandomised}. Power improvements for e-BH have also been studied through randomization and conditional calibration~\citep{xu2023more,lee2024boosting}. Our fixed-sample results sharpen the calibration of e-value rejection thresholds using additional regularity of the null distribution.

The sequential setting requires additional care. Nonnegative supermartingales provide a classical route to anytime-valid inference, with Ville's inequality controlling their running maxima~\citep{ville1939etude,howard2020time}. Modern work connects these martingale ideas to e-processes and anytime-valid inference~\citep{ramdas2020admissible,ramdas2023game}, and develops safe testing, powerful e-statistics, and methods for combining sequential e-values~\citep{grunwald2024safe,perez2024statistics,vovk2024merging}. Randomized Ville calibration retains the standard boundary during monitoring and adds an independent randomized check at the terminal stopping time~\citep{ramdas2026randomized}. In our sequential construction, Lipschitz regularity is used to construct new conditional e-value sequence whose product forms a valid supermartingale.

Most directly related, \citet{blier2024improved} obtain approximately factor-two small-$\alpha$ threshold reductions for decreasing or unimodal densities and also study boosting e-values for e-BH under additional distributional assumptions. Motivated by the Markov worst case, we study a nondecreasing density to restrict tail allocation and a Lipschitz density to control local concentration. This use of Lipschitz regularity follows the standard nonparametric idea of describing a density through a smoothness class~\citep{tsybakov2009introduction}.

\section{Problem Setup} \label{sec:setup}

Let $E$ be a nonnegative e-variable satisfying $\mathbb{E}[E] \le 1$ under the null hypothesis. For $\alpha \in (0,1)$ and a nonempty class $\mathcal E$ of admissible e-variables, define the smallest rejection threshold that controls the type-I error at significance level $\alpha$ by
\begin{equation*}
    T_\alpha(\mathcal{E})
    \triangleq
    \inf\left\{
        t\geq0:
        \sup_{E\in\mathcal{E}} \mathbb{P}(E \geq t) \leq \alpha
    \right\}.
\end{equation*}
For unrestricted e-variables, Markov's inequality gives $T_\alpha=1/\alpha$. Under additional structure on the null distribution, smaller valid thresholds may be possible. We study two such structural restrictions.

\begin{definition}[Nondecreasing-density class]
    Let $b>0$ be a fixed support bound. The class $\mathcal E_b$ consists of e-variables whose null distribution is supported on $[0,b]$, allows an atom at zero, and admits a non-decreasing density $f$ on $(0,b]$.
\end{definition}

\begin{definition}[Lipschitz-density class]
    Let $L>0$ be a known bound. The class $\mathcal E_L$ consists of e-variables whose null distribution admits a density $f$ on $[0, \infty)$ satisfying
    \begin{equation*}
        |f(x)-f(y)| \leq L|x-y|,\quad x,y \geq 0.
    \end{equation*}
\end{definition}
Here $L$ is an upper bound, not necessarily the smallest Lipschitz constant. We next introduce the two standard e-value inference procedures used in our later developments: \textbf{e-BH} for multiple testing and \textbf{e-processes} for sequential inference.

For multiple testing, consider $K$ hypotheses with e-values $E_1,\ldots,E_K$, and write $\mathcal K=\{1,\ldots,K\}$ and $\mathcal H_0\subseteq\mathcal K$ for the true-null indices. Let $e_{(1)}\geq\cdots\geq e_{(K)}$ denote the ordered realizations. Given deterministic thresholds $t_1\geq\cdots\geq t_K$, the step-up procedure rejects the hypotheses corresponding to the $R$ largest e-values, where
\begin{equation}
\label{eq:def-R}
    R\triangleq\max\left\{k\in\mathcal K:e_{(k)}\geq t_k\right\},
\end{equation}
with $R=0$ if the set is empty. Let $\mathcal R$ denote the rejection set. The false discovery rate is
\begin{equation*}
    \mathrm{FDR}\triangleq \mathbb E \left[ \frac{|\mathcal R \cap\mathcal H_0|}{R \vee 1} \right],
\end{equation*}
and the target is $\mathrm{FDR}\leq\alpha$. Standard e-BH uses $t_k=K/(k\alpha)$~\citep{wang2022false}.

For sequential inference, let $(\mathcal F_t)_{t\geq0}$ be a filtration. A nonnegative, adapted sequence $(E_t)_{t\geq1}$ is called a sequence of conditional e-values if $\mathbb E[E_t\mid\mathcal F_{t-1}]\leq1$ under the null hypothesis. The cumulative product $M_t=\prod_{s=1}^t E_s$ is a nonnegative supermartingale with $M_0=1$~\citep{vovk2024merging}. At level $\alpha\in(0,1)$, the sequential test rejects at
\begin{equation*}
    \tau_\alpha\triangleq\inf\{t\geq1:M_t\geq1/\alpha\},
\end{equation*}
where $\inf\emptyset=\infty$. Ville's inequality guarantees $\mathbb P (\tau_\alpha < \infty) \le \alpha$~\citep{ville1939etude,howard2020time}.

\section{Rejection Thresholds under Structural Constraints} \label{sec:Lip-constrained thresholds}

Markov's inequality uses only the first-moment bound $\mathbb E[E]\leq1$. Its extremal distribution places all tail mass exactly at the threshold and the remaining mass at zero; the mean then equals the threshold multiplied by the tail probability. The two structural classes in~\Cref{sec:setup} restrict different parts of this construction. For $\mathcal E_b$, nondecreasingness forces tail mass to extend towards the support bound. For $\mathcal E_L$, Lipschitz regularity forces mass near zero to occupy an interval of positive width. We develop the two calibrations separately: the first admits an exact minimax characterization, while the second gives an explicit threshold with additive asymptotic sharpness.

\subsection{Thresholds under Nondecreasing Densities} \label{subsec:shape}

In the class $\mathcal E_b$, we examine whether the density constraint permits a lower threshold within the support, and how far it can be reduced while keeping the rejection probability at most $\alpha$ for every null distribution in the class.

\begin{theorem}[Nondecreasing-density thresholds and minimax optimality]
\label{thm:mi-benchmark}
Let $\alpha\in(0,1)$ and $b>0$. For $0<t<b$,
\begin{equation*}
    \sup_{E\in\mathcal E_b}\mathbb P(E\ge t)
    =\min\left\{1, \frac{2}{t+b}\right\},
\end{equation*}
whereas $\mathbb P(E \ge b) = 0$ for every $E \in \mathcal E_b$. The minimax threshold infimum is
\begin{equation*}
    T_\alpha(\mathcal E_b)=
    \begin{cases}
        \min\{b,\,2/\alpha-b\}, & b<2/\alpha,\\
        0, & b\ge2/\alpha.
    \end{cases}
\end{equation*}
When $b<2/\alpha$, the threshold $T_\alpha^{(b)}=\min\{b,2/\alpha-b\}>0$ is uniformly valid and minimax optimal. When $b\ge2/\alpha$, every strictly positive threshold is uniformly valid.
\end{theorem}

\begin{proof}[\textbf{Proof Sketch}]
For $0<t<b$, Chebyshev's integral inequality gives
\begin{equation*}
    \begin{aligned}
        1&\ge\mathbb E[E]
        \ge\int_t^b xf(x)\,\mathrm dx\\
        &\ge\frac{1}{b-t}
        \int_t^b x\,\mathrm dx
        \int_t^b f(x)\,\mathrm d x =\frac{t+b}{2}\mathbb P(E\ge t).
    \end{aligned}
\end{equation*}
The upper bound is attained by placing mass $\min\{1,2/(t+b)\}$ uniformly on $[t,b]$ and the rest at zero. The resulting density is nondecreasing on $(0,b]$, and its mean is at most one. This construction establishes the minimax optimality of the positive threshold stated in the theorem. Appendix~\ref{app:monotone-benchmark} provides the full proof.
\end{proof}

The first-moment calculation also identifies the source of the improvement. For $0<t<b$, nondecreasingness yields
\begin{equation*}
    \mathbb E[(E-t) \mathbf 1 \{E \ge t\}] \ge \frac{b-t}{2} \mathbb P(E\ge t).
\end{equation*}
Thus each unit of tail probability contributes at least $(t+b)/2$ to the mean, rather than the $t$ used in Markov's bound. Probability mass at zero contributes zero to the mean. The sharper calibration in this class is therefore driven by the excess first moment within the tail.

\subsection{Thresholds under Lipschitz Densities}

The class $\mathcal E_L$ supplies a different kind of information: it bounds the density's local variation while allowing nonmonotone shapes and unbounded support. Its calibration relies on the probability below the rejection threshold. A fixed amount of mass cannot be packed arbitrarily close to zero under a slope bound, and moving that mass away from zero forces it to contribute to the mean. The following lemma quantifies the smallest possible contribution.

\begin{lemma}[Lipschitz integral inequality] \label{lemma:lipschitz}
    Let $L,t>0$ and let $f:[0,t]\to\mathbb R_{\ge0}$ be $L$-Lipschitz. If $m=\int_0^t f(x)\,\mathrm dx$ and $t\ge\sqrt{2m/L}$, then
    \begin{equation*}
        \int_0^t xf(x)\,\mathrm dx
        \ge\frac{\sqrt2}{3\sqrt L}m^{3/2}.
    \end{equation*}
    Equality holds exactly for $f(x)=(\sqrt{2Lm}-Lx)_+$.
\end{lemma}

The equality case is a triangular density profile descending from zero with the largest permitted slope. It places the prescribed mass as far to the left as the Lipschitz condition allows, and its support has width $\sqrt{2m/L}$. The lemma therefore turns a bound on local concentration into a lower bound on the first moment. Applied below a threshold satisfying $t\ge\sqrt{2\mathbb P(E<t)/L}$, it gives
\begin{equation*}
    \mathbb E[E\mathbf1\{E<t\}]
    \ge\frac{\sqrt2}{3\sqrt L}\mathbb P(E<t)^{3/2}.
\end{equation*}
Combining this contribution with the tail moment yields
\begin{equation*}
    1\ge\mathbb E[E]
    \ge\frac{\sqrt2}{3\sqrt L}\mathbb P(E<t)^{3/2}
    +t\mathbb P(E\ge t).
\end{equation*}
At the target rejection probability $\alpha$, the mass below the threshold is $1-\alpha$. Accounting for its required first moment leads to the following explicit calibration.

\begin{theorem}[Lipschitz threshold]\label{thm:lipschitz_threshold}
Let $\alpha\in(0,1)$ and $L\ge4/9$. The threshold
\begin{equation*}
    T_\alpha^{(L)}
    =\frac{1}{\alpha}
    \left(1-\frac{\sqrt2(1-\alpha)^{3/2}}{3\sqrt L}\right)
    \in(0,1/\alpha)
\end{equation*}
satisfies $\sup_{E\in\mathcal E_L}\mathbb P(E\ge T_\alpha^{(L)})\leq\alpha$.
\end{theorem}

\begin{proof}[Proof Sketch]
Set $t = T_\alpha^{(L)}$ and suppose $\mathbb P(E\ge t) > \alpha$. $L \ge 4 / 9$ ensures $t \ge \sqrt{ \frac{2 (1-\alpha)}{L}} \ge \sqrt{\frac{2 \mathbb P(E<t)}{L}}$,
so the preceding moment bound applies. Rearranging it gives
\begin{equation*}
    t\leq
    \frac{1-\sqrt2\,\mathbb P(E<t)^{3/2}/(3\sqrt L)}
    {\mathbb P(E\ge t)}.
\end{equation*}
The function $u\mapsto[1-\sqrt2(1-u)^{3/2}/(3\sqrt L)]/u$ is strictly decreasing on $(0,1]$. Since $\mathbb P(E\ge t)>\alpha$, the right-hand side is strictly smaller than $T_\alpha^{(L)}$, a contradiction. The range and monotonicity calculations, together with the proof of Lemma~\ref{lemma:lipschitz}, appear in Appendix~\ref{app:threshold-proof}.
\end{proof}

The correction subtracts the first moment forced below the threshold before assigning the remaining budget to the tail. It depends explicitly on both the rejection level and the density bound:
\begin{equation*}
    1-\alpha T_\alpha^{(L)}
    =\frac{\sqrt2(1-\alpha)^{3/2}}{3\sqrt L}.
\end{equation*}
A smaller valid bound $L$ gives a larger correction, while the threshold approaches $1/\alpha$ as $L$ increases. $L \ge 4 / 9$ is sufficient for the lemma's interval requirement to hold at the proposed threshold for every $\alpha \in (0,1)$. 

The proposed threshold may not be exactly minimax optimal for a fixed $\alpha$. However, the following theorem shows that, for each fixed $L>4/9$, its difference from the minimax threshold tends to zero as $\alpha\downarrow0$.

\begin{theorem}[Small-$\alpha$ asymptotic sharpness]\label{thm:lip_property}
Fix $L>4/9$. For every $\varepsilon>0$ and all sufficiently small $\alpha>0$, there exists $E\in\mathcal E_L$ such that
\begin{equation*}
    \mathbb P\bigl(E\ge T_\alpha^{(L)}-\varepsilon\bigr)>\alpha.
\end{equation*}
Consequently,
\begin{equation*}
    0 \le T_\alpha^{(L)} - T_\alpha (\mathcal E_L) \longrightarrow 0
    \qquad (\alpha \downarrow 0).
\end{equation*}
\end{theorem}

The proof constructs a null distribution in $\mathcal E_L$ from a triangular component near zero and a narrow triangular tail, each with slopes bounded in absolute value by $L$. The first component attains the lemma's lower bound for the mass below the threshold. The tail is centered so that the total mean is one, and its width shrinks with its probability mass. With tail mass $\alpha$, its left endpoint is $T_\alpha^{(L)}-\sqrt{\alpha/L}$. A small increase in the tail mass then makes any proposed fixed additive reduction invalid once $\alpha$ is sufficiently small. The explicit construction and its proof are given in Appendix~\ref{app:sharpness-proof}.

At the same time, the relative improvement over Markov calibration remains nonzero:
\begin{equation*}
    1-\alpha T_\alpha^{(L)}
    \longrightarrow\frac{\sqrt2}{3\sqrt L}
    \qquad(\alpha\downarrow0).
\end{equation*}
The distinction between these two scales is useful. The Lipschitz restriction removes a fixed fraction of the Markov threshold as rejection becomes rare, while any further uniformly valid improvement over the proposed threshold must tend to zero in absolute size. The persistent relative gain arises because the probability below the threshold approaches one and continues to require a positive first moment.

\textbf{Unified view.}
Both calibrations refine the decomposition
\begin{equation*}
    \mathbb E[E] = t\mathbb P(E\ge t) + \mathbb E[E\mathbf1\{E<t\}] + \mathbb E[(E-t) \mathbf 1 \{E\ge t\}].
\end{equation*}
Markov's inequality retains only the first term. Nondecreasingness supplies a lower bound on the excess moment within the tail, while the Lipschitz argument supplies a lower bound on the moment below the threshold. These contributions are zero for Markov's extreme value distribution. Together, these bounds show that structural constraints make a given tail probability more costly in terms of the first moment, which leads to a smaller valid rejection threshold.

\section{Applications of the Lipschitz Calibration} \label{sec:applications}

The Lipschitz threshold refines the Markov calibration for structured e-values \citep{vovk2021values}. We show that the same structural calibration enters both multiple testing and sequential inference: in multiple testing, it gives a uniformly sharper replacement for the e-BH threshold \citep{wang2022false}; in sequential inference, it yields an explicit update interface for constructing anytime-valid processes \citep{ville1939etude, howard2020time, ramdas2023game}.

\subsection{Multiple Testing}

We first incorporate the Lipschitz threshold into the standard step-up framework~\citep{benjamini1995controlling}.  Since e-BH rejects by comparing ordered e-values to Markov-based thresholds~\citep{wang2022false}, a natural modification is to replace those thresholds by their Lipschitz-calibrated counterparts. The resulting procedure denoted e-LC rejects the hypotheses corresponding to the $R$ largest e-values, with $R$ defined in Eq.~\eqref{eq:def-R} and $t_k \triangleq T_{k\alpha/K}^{(L)}$.
The next proposition establishes FDR control when the inputs are independent.

\begin{proposition}[Lipschitz-calibrated e-BH under independence] \label{prop:lipschitz_fdr}
    Fix $\alpha\in(0,1)$ and $L\ge4/9$. Let $\mathcal H_0$ denote the set of true nulls and let $K_0=|\mathcal H_0|$. Assume that the e-values $(E_1,\ldots,E_K)$ are mutually independent and that, for each $i\in\mathcal H_0$, the null distribution of $E_i$ belongs to $\mathcal{E}_L$. Then the above step-up procedure satisfies
    \begin{equation}
        \mathrm{FDR}\leq \frac{K_0}{K}\alpha\leq \alpha.
    \end{equation}
\end{proposition}

The proof inverts the threshold curve to obtain induced p-values and applies BH to them. Monotonicity and uniform validity of $T_u^{(L)}$ at each level $u$ make these p-values super-uniform under their nulls; independence is inherited from the original e-values. Algorithm~\ref{alg:elc-stepup} gives the direct implementation. The procedure requires \(O(K\log K)\) time due to sorting and \(O(K)\) space.

\begin{algorithm}[h]
\caption{e-LC procedure}
\label{alg:elc-stepup}
\KwIn{Nonnegative e-values $(e_i)_{i=1}^K$; level $\alpha\in(0,1)$; bound $L\ge4/9$.
}
\KwOut{Rejected original indices $\mathcal R\subseteq \mathcal K$ and count $R = |\mathcal R|$.}
Sort indices $\pi(1),\ldots,\pi(K)$ so that $e_{\pi(1)}\ge\cdots\ge e_{\pi(K)}$ (break ties by original index)\;
$R \gets 0$\;
\For{$k=1$ \KwTo $K$}{
  $u \gets \alpha k/K$\;
  $t_k \gets \bigl(1-\sqrt{2}(1-u)^{3/2}/(3\sqrt L)\bigr)/u$\;
  \If{$e_{\pi(k)} \ge t_k$}{
    $R \gets k$\;
  }
}
$\mathcal R \gets \{\pi(j):1\le j\le R\}$ \tcp*{Empty if $R=0$}
\Return{$(\mathcal R,R)$}\;
\end{algorithm}

To retain the arbitrary-dependence guarantee of standard e-BH, we use the regularity information to transform the inputs themselves. The lower-tail moment bound limits the null expectation contributed by e-values above $T_\alpha^{(L)}$. Truncating below this threshold and rescaling the retained values therefore produces valid boosted e-values, to which the usual e-BH guarantee applies.

\begin{theorem}[Boosted e-LC under arbitrary dependence] \label{thm:lipschitz-truncated-ebh}
    Let $\alpha \in (0,1)$, $L \ge 4 / 9$, and $d_{\alpha,L} = 1 - \sqrt 2 (1 - \alpha)^{3 / 2} / (3 \sqrt L)$. Given $\{E_k\}_{k \in \mathcal K}$, define
    \begin{equation*}
        E_i^{\mathfrak b} = \frac{E_i}{d_{\alpha, L}} \mathbf 1 \{E_i \ge T_\alpha^{(L)}\}.
    \end{equation*}
    If every true-null $E_i$ belongs to $\mathcal E_L$, standard e-BH at level $\alpha$ applied to $E_1^{\mathfrak b},\ldots,E_K^{\mathfrak b}$ satisfies $\mathrm{FDR}\le(K_0/K)\alpha\le\alpha$ under arbitrary dependence. Equivalently, the original-scale step-up thresholds are $t_k^{\mathfrak b} = d_{\alpha,L}K/(\alpha k)$, $k \in \mathcal K$.
\end{theorem}

\begin{proof}[Proof Sketch]
    For each true null $i$, the threshold and lower-tail moment bounds imply $\mathbb E_0[E_i\mathbf1\{E_i\ge T_\alpha^{(L)}\}]\le d_{\alpha,L}$. Dividing by $d_{\alpha,L}$ makes $E_i^{\mathfrak b}$ a valid null e-value. Standard e-BH then gives the stated FDR guarantee~\citep{wang2022false}. Since $d_{\alpha,L}K/(\alpha k)\geq T_\alpha^{(L)}$, its rejection condition on the boosted scale is equivalent to the stated threshold on the original scale. Appendix~\ref{app:lip-truncated-ebh-proof} gives the complete argument.
\end{proof}

The two e-LC variants use the same Lipschitz correction in different ways. The direct rule applies the calibrated threshold at each rank level $k\alpha/K$, whereas the boosted rule uses the common normalization $d_{\alpha,L}$ to preserve e-value validity. Their thresholds satisfy
\begin{equation*}
    T_{k\alpha/K}^{(L)}\leq t_k^{\mathfrak b}\leq\frac{K}{k\alpha}, \quad k \in \mathcal K.
\end{equation*}
Thus both procedures improve on standard e-BH~\citep{wang2022false}. The direct rule retains the sharper rank-dependent thresholds with FDR control under independence, while the boosted rule trades some of this sharpness for validity under arbitrary dependence.

\subsection{Sequential Inference}

Sequential validity requires a different use of the Lipschitz condition. The fixed-sample threshold $T_\alpha^{(L)}$ controls the tail probability of a single e-value, whereas sequential inference requires time-uniform control of a cumulative process. We therefore keep Ville's boundary $1/\alpha$ unchanged and use Lipschitz regularity to calibrate each one-step update. Specifically, we transform $E_t$ into $\widetilde E_t=\phi_t(E_t)$ and form $M_t=\prod_{s=1}^t\widetilde E_s$. Let $(\mathcal F_t)$ denote the underlying filtration, with $E_t$ measurable with respect to $\mathcal F_t$. For a nonnegative update $\phi$, define $\Psi_L(\phi)=\sup_{E\in\mathcal E_L}\mathbb E[\phi(E)]$.

\begin{lemma}[Supermartingale Property of Inference Function] \label{lemma:supermartingale}
    If $\phi_t$ is a nonnegative predictable update, $\Psi_L(\phi_t)\le1$, and $E_t\mid\mathcal F_{t-1}$ has distribution in $\mathcal E_L$, then $M_0=1$ and $M_t=\prod_{s=1}^t\phi_s(E_s)$ form a nonnegative supermartingale.
\end{lemma}

We now seek a simple update family whose worst-case expectation can be calibrated explicitly. Consider the thresholded linear update
\begin{equation*}
\label{equation:lcupdate}
    \phi_{\lambda,\eta,\tau}(x)=(1-\lambda)+\lambda\eta(x-\tau)_+,
\end{equation*}
where $\lambda\in[0,1]$, $\eta\geq0$, and $\tau\geq0$. The parameter $\lambda$ controls the aggressiveness of the update, while $\tau$ determines which part of the input contributes to evidence accumulation.

We next derive an explicit calibration of $\eta$. This requires controlling the worst-case value of $\mathbb E[(E-\tau)_+]$. To obtain a closed-form bound, we impose the additional boundary condition $f(0)=0$. Together with Lipschitz regularity, this condition limits the probability mass near zero and leads to the following calibration.

\begin{theorem}[Closed-form Sequential Calibration] \label{thm:process-based}
    Let $L>0$, and let $\mathcal{E}_L^0$ denote the subclass of $\mathcal{E}_L$ whose densities satisfy $f(0)=0$. Suppose that $\mathcal{E}_L^0$ is nonempty, which necessarily implies $L\ge 8/9$. Let $\Psi_L^0$ denote the corresponding worst-case expectation operator over $\mathcal{E}_L^0$. For any threshold $\tau \in [0, 1/\sqrt{2L}]$, define
    \begin{equation}
        \eta_L(\tau) \triangleq \frac{1}{1 - \tau + \frac{1}{3}L\tau^3}. 
    \end{equation}
    Then for every strategy parameter $\lambda \in [0,1]$, this calibration ensures
    $\Psi_L^0(\phi_{\lambda,\eta_L(\tau),\tau}) \le 1$.

     Consequently, for any predictable sequence $\phi_t=\phi_{\lambda_t,\eta_L(\tau_t),\tau_t}$ with $\lambda_t\in[0,1]$ and $0\leq\tau_t\leq1/\sqrt{2L}$, if the conditional null law of each $E_t$ belongs to $\mathcal E_L^0$ and satisfies $\mathbb E[E_t\mid\mathcal F_{t-1}]\leq1$, then $M_0=1$ and $M_t=\prod_{s=1}^t\phi_s(E_s)$ form a nonnegative supermartingale. Hence, for every $\alpha\in(0,1)$,
    \begin{equation*}
        \mathbb{P}\left( {\sup}_{t \geq 1} M_t \geq {1}/{\alpha} \right) \le \alpha.
    \end{equation*}
\end{theorem}

\begin{proof}[Proof Sketch]
    The boundary condition gives $f(x)\le Lx$ and bounds the mass below $\tau$. Combining this with Lemma~\ref{lemma:lipschitz} yields $\sup_{E\in\mathcal E_L^0}\mathbb E[(E-\tau)_+]\le1-\tau+L\tau^3/3$. Its reciprocal calibration therefore ensures $\Psi_L^0(\phi_t)\le1$. Conditional expectation and Ville's inequality give the result. The moment calculation and positivity check are in Appendix~\ref{app:process-based-proof}.
\end{proof}

\section{Numerical Experiments}
\label{sec:experiments}

\begin{figure*}[t]
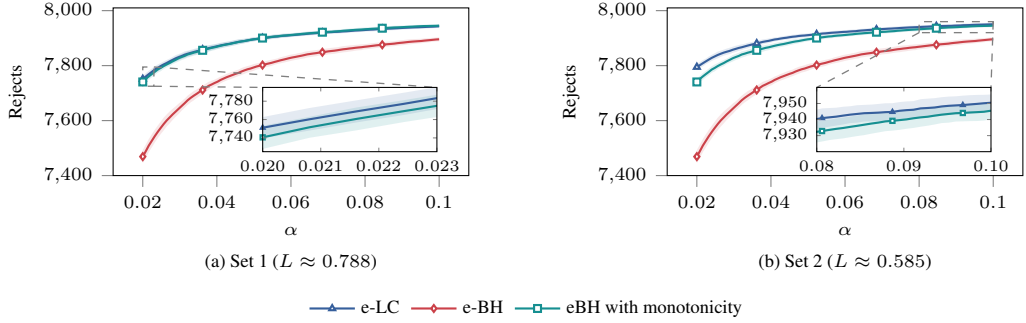

    \centering
    \input{figs/tikz/fdr_elc1}%
    \hfill
    \input{figs/tikz/fdr_elc2}
\par\smallskip\pgfplotslegendfromname{fdrlegend}
    \caption{Mean rejection counts under two representative Lipschitz settings. Bands show SD over 20 repetitions.}
    \label{fig:lc_comparison}
\end{figure*}

We have shown that Lipschitz calibration yields \textit{ sharper thresholds} in fixed-sample multiple testing (\Cref{thm:lipschitz_threshold}) and \textit{valid one-step updates} for sequential testing (\Cref{thm:process-based}). In this section, we examine whether these calibrations lead to more discoveries and detections. All experiments were run on standard CPU hardware; no GPU acceleration was required.

\subsection{Power Gain in Multiple Testing}
\label{subsec:exp_multiple}

We first show how the threshold gain from Lipschitz calibration translates into more discoveries in fixed-sample multiple testing. We compare e-LC with standard e-BH and the benchmark from \citet{blier2024improved} using the same e-values. Our null e-values follow the normalized two-segment piecewise-linear density
\begin{equation}
f_0(x)=c
\begin{cases}
1-k_1x, & 0\le x\le x_1,\\
1-k_1x_1-k_2(x-x_1), & x_1<x\le x_2,\\
0, & \text{otherwise},
\end{cases}
\label{equ:family}
\end{equation}

where $c$ is the normalizing constant. Notice this density is both \textit{nonincreasing} and \textit{$L$-Lipschitz} with $L=c\max\{k_1,k_2\}$. At rank level $u=k\alpha/K$, the Lipschitz condition gives the e-LC threshold in \Cref{thm:lipschitz_threshold}, while the nonincreasing-density condition gives the threshold $1/(2u)$ of \citet{blier2024improved}. Because both structural calibrations apply to the same null family, $1/(2u)$ provides a direct benchmark for isolating the gain from Lipschitz calibration. All three procedures are applied to the same e-values and evaluated over 100 equally spaced levels $\alpha\in[0.02,0.10]$. Full configuration details are provided in Appendix~\ref{app:exp_setup}.

Comparing $T_u^{(L)}$ in \Cref{thm:lipschitz_threshold} with $1/(2u)$, e-LC is more aggressive whenever $L<\frac{8}{9}(1-u)^3$. This condition holds broadly across feasible settings in the piecewise-linear family. We performed a randomized feasibility study retaining only valid e-value distributions where e-LC has a lower threshold than the monotonicity benchmark in $72.82\%$, $81.88\%$, and $85.21\%$ of the feasible configurations at $\alpha\in\{0.05,0.02,0.01\}$, respectively. We, therefore, illustrate the $L^{-1/2}$ correction in \Cref{thm:lipschitz_threshold} using two representative examples. In \Cref{fig:lc_comparison}(a), where $L\approx0.788$, the two thresholds are close over the relevant range of $u$, and e-LC performs similarly to the monotonicity benchmark. In \Cref{fig:lc_comparison}(b), where $L\approx0.585$, the correction is larger, and e-LC yields more rejections than the monotonicity benchmark throughout the tested levels. Thus, smaller $L$ produces a larger threshold reduction, which translates into larger discovery gains in this experiment.

\subsection{Sequential and Anytime-Valid Inference}
\label{subsec:exp_sequential}

\begin{figure*}[t]
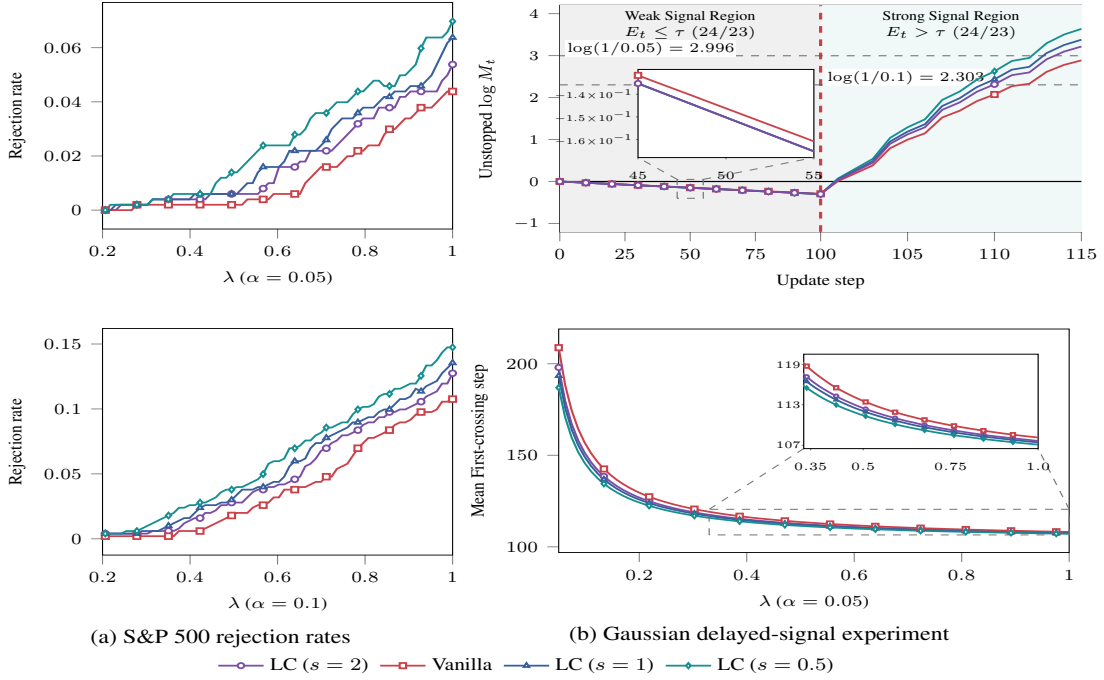

\centering

\adjustbox{width=\textwidth,height=.20\textheight}{%
\begin{minipage}{\textwidth}
\centering
\begin{minipage}[t]{.425\textwidth}
    \vspace{0pt}
    \centering
    \input{figs/tikz/sp500_rejection_rates}
    \vspace{-.3em}
    {\small (a) S\&P 500 rejection rates}
    \label{fig::sp500rejection}
\end{minipage}\hfill
\begin{minipage}[t]{.555\textwidth}
    \vspace{0pt}
    \centering
    \input{figs/tikz/mechanistic_demo}
    \par\smallskip
    \makebox[\linewidth][c]{%
    \hspace*{0.35cm}%
    \input{figs/tikz/gaussian_mean_first_crossing}%
    }
    \vspace{-.4em}
    {\small (b) Gaussian delayed-signal experiment}
    \label{fig::gussianexp}
\end{minipage}

\par\smallskip
\pgfplotslegendfromname{sequentiallegend}
\end{minipage}%
}

\caption{Sequential LC experiments on S\&P 500 paths and Gaussian delayed signals.}
\label{fig:sequential-mechanism}
\end{figure*}

We then investigate Lipschitz calibration in sequential testing via a controlled Gaussian mechanistic demonstration and a retrospective S\&P 500 analysis. The Gaussian demonstration uses a single shared likelihood-ratio e-value path $E_t=\exp(\theta X_t-\theta^2/2)$ with $\theta=0.3$. The weak-signal phase is constructed so that $E_t<\tau_L$ for every LC calibration, while the common strong-signal phase is constructed so that $E_t>24/23$. All methods use the same realized inputs in both phases. We compare $L=cL_0$ with $c\in\{0.5,1,2\}$ and directly verify the Gaussian null expectation for each LC calibration. The S\&P 500 experiment uses adjusted-close prices from January 2017 to June 2026, where each stock forms a sequential path and consecutive selected prices define the ratio inputs $E_{k,j}^{\mathrm{base}}=P_{k,j}/P_{k,j-1}$. For each path, we estimate a working Lipschitz scale from a fitted lognormal density and evaluate the same scale multipliers. The stock experiment is retrospective because the conditional e-value and conditional density assumptions are not established by the historical data.

Selective LC calibration can convert stronger late evidence into earlier Ville-boundary crossings. To preserve the null-expectation constraint, LC suppresses weak observations through $(E_t-\tau)_+$ and uses the calibrated scaling $\eta_L(\tau)$ to place more weight on sufficiently strong observations. Under our choice $\tau=0.5/\sqrt{2L}$, the calibrated evidence term is $\eta_L(\tau)(E_t-\tau)_+$, and it exceeds the vanilla term $E_t$ whenever $E_t>24/23$. Thus observations below this crossover are suppressed relative to vanilla, while observations above it lie in the LC amplification regime. The resulting factor of \Cref{equation:lcupdate} therefore remains valid while allowing larger multiplicative increments when the observed evidence is strong. As shown in the upper panel of \Cref{fig::gussianexp}, the weak-signal phase is constructed such that $E_t<\tau_k$ for all LC calibrations, so $(E_t-\tau_k)_+=0$ and all LC variants use the common factor $1-\lambda$, slightly below the vanilla update. In the strong-signal phase, $E_t>24/23$ by construction, so the LC update exceeds the vanilla update, causing faster growth and earlier crossings of both Ville boundaries. The lower panel further shows that LC has a smaller mean first-crossing time than vanilla across the tested range of $\lambda$. Earlier crossings increase the chance that a path rejects before the monitoring horizon, which can translate into a higher rejection rate. Consistently, \Cref{fig::sp500rejection} shows higher rejection rates for all LC variants than vanilla.

It should be noted that although \Cref{subsec:exp_multiple} and \Cref{subsec:exp_sequential} both show that gains from Lipschitz calibration can translate into more rejections or detections, the two setting, however, use Lipschitz information differently and therefore the source of the gains are different. In \Cref{subsec:exp_multiple}, the gain comes from a lower rejection threshold established in \Cref{thm:lipschitz_threshold}. Since the same e-values are compared against a smaller threshold than in standard e-BH, more hypotheses can satisfy the rejection condition and be discovered. Furthermore, the relative threshold reduction is of order $O(L^{-1/2})$, so a smaller $L$ allows a larger reduction and potentially larger discovery gains. In \Cref{subsec:exp_sequential}, the gain comes from the multiplicative update applied to each evidence term. The vanilla update uses the factor $(1-\lambda)+\lambda E$, while the LC update consider $(1-\lambda)+\lambda\eta_L(\tau)(E-\tau)_+$. Under the conditions of \Cref{thm:process-based} and with both processes compared against the same Ville boundary, this calibrated update preserves the supermartingale validity, while  $\eta_L(\tau)(E-\tau)_+$ can exceed the vanilla term $E$ for sufficiently strong observations. The cumulative LC process can therefore grow faster along favorable paths and cross the same Ville boundary more often, which provides a mechanism for the higher rejection rates observed in \Cref{subsec:exp_sequential}.

\section{Discussion} \label{sec:discussion}

This work starts from the worst-case construction behind Markov's inequality and asks how much the rejection threshold can improve once different parts of that construction are restricted. A nondecreasing density limits how tail mass can be placed and gives an exact threshold characterization. A Lipschitz density limits sharp local concentration and yields an explicit improvement with small-$\alpha$ asymptotic sharpness. We further incorporate Lipschitz information into both e-BH procedures and anytime-valid sequential inference. These results show how modest distributional information can be translated into sharper and more broadly useful e-value procedures.

Future work includes developing valid data-driven upper bounds for $L$ from empirical or prior information, beyond settings where $L$ can be derived analytically. Another direction is to develop practical guidelines for choosing sequential strategies under different data settings, signal patterns, and inferential goals.

\clearpage
\raggedbottom
\bibliographystyle{apalike}
\begingroup
\setlength{\bibsep}{0pt}
\bibliography{Reference}
\endgroup

\clearpage
\appendix
\flushbottom
\section{Proofs}
\label{app_prof}

\subsection{Proof of \texorpdfstring{\Cref{thm:mi-benchmark}}{Theorem \getrefnumber{thm:mi-benchmark}}}
\label{app:monotone-benchmark}

\begin{proof}
    First, $\mathcal F_b$ is nonempty for every $b>0$: with $p=\min\{1,2/b\}$, the mixture $(1-p)\delta_0+p\,\mathrm{Unif}(0,b)$ has mean $pb/2\leq1$ and a constant density on $(0,b]$.

    Fix $t\in(0,b)$ and write $p_t\triangleq\mathbb P(E\geq t)=\int_t^b f(x)\,\mathrm dx$. Since both $x$ and $f(x)$ are non-decreasing on $[t,b]$, Chebyshev's integral inequality gives
    \begin{align*}
        1\geq\mathbb E[E]
        &\geq\int_t^b xf(x)\,\mathrm dx\\
        &\geq\frac{1}{b-t}\left(\int_t^b x\,\mathrm dx\right)
        \left(\int_t^b f(x)\,\mathrm dx\right)
        =\frac{t+b}{2}p_t.
    \end{align*}
    Together with $p_t\leq1$, this yields
    \begin{equation*}
        \sup_{E\in\mathcal F_b}\mathbb P(E\geq t)
        \leq\min\left\{1,\frac{2}{t+b}\right\}.
    \end{equation*}

    To attain this bound, set $q=\min\{1,2/(t+b)\}$ and consider
    \begin{equation*}
        P_E^*=(1-q)\delta_0+q\,\mathrm{Unif}[t,b].
    \end{equation*}
    Its positive continuous part has density $f^*(x)=q(b-t)^{-1}\mathbf1\{t\leq x\leq b\}$ on $(0,b]$, which is non-decreasing and has mass $q$. Moreover, $\mathbb E_{P_E^*}[E]=q(t+b)/2\leq1$ and $\mathbb P_{P_E^*}(E\geq t)=q$. Hence the tail bound is exact for every $t\in(0,b)$.

    At $t=b$, the tail probability is zero because the positive part is absolutely continuous and has no atom at $b$; it also vanishes for $t>b$. At $t=0$, nonnegativity gives $\mathbb P(E\geq0)=1$ for every member of the class.

    It remains to derive the threshold infimum. For $0<t<b$, the exact worst-case tail probability is at most $\alpha$ if and only if $2/(t+b)\leq\alpha$, equivalently $t\geq2/\alpha-b$. If $b\leq1/\alpha$, then $2/\alpha-b\geq b$, so no threshold in $(0,b)$ is valid and the smallest valid threshold is $b$. If $1/\alpha<b<2/\alpha$, then $2/\alpha-b\in(0,b)$ is the smallest valid threshold. Finally, if $b\geq2/\alpha$, then every $t\in(0,b)$, and therefore every $t>0$, is valid, so the threshold infimum is zero. These three cases give the stated expression.

    In the last regime, the infimum is not attained by the rule $E\geq t$, since $\mathbb P(E\geq0)=1>\alpha$. As $t\downarrow0$, the events $\{E\geq t\}$ increase to $\{E>0\}$, so continuity of probability and the tail bound give
    \begin{equation*}
        \mathbb P(E>0)\leq\min\left\{1,\frac{2}{b}\right\}.
    \end{equation*}
    The mixture $(1-p)\delta_0+p\,\mathrm{Unif}(0,b)$ from the first paragraph attains this bound, since its positive mass is $p=\min\{1,2/b\}$. Thus $E>0$ is a valid rejection rule when $b\geq2/\alpha$.
\end{proof}

\subsection{Proof of \texorpdfstring{\Cref{thm:lipschitz_threshold}}{Theorem \getrefnumber{thm:lipschitz_threshold}}}
\label{app:threshold-proof}
In this section, we provide the detailed proofs for the Lipschitz-constrained results presented in \Cref{thm:lipschitz_threshold}. We begin with a technical lemma establishing the minimum first moment of a Lipschitz function, which forms the basis for the main theorem.

\begin{proof}[Proof of \Cref{lemma:lipschitz}]
{If $m=0$, nonnegativity and continuity imply $f\equiv0$, so the inequality and equality characterization hold. Hence assume $m>0$. Let}
\[
h=\sqrt{2Lm},\qquad n=\frac{h}{L}=\sqrt{\frac{2m}{L}}\le t,
\qquad s=\int_n^t f(x)\,\mathrm{d}x,
\]
and define $f_1=f+s/n$ on $[0,n]$ and $f_1=0$ on $(n,t]$.
Then $\int_0^n f_1=m$, $f_1$ is $L$-Lipschitz on $[0,n]$, and
\begin{equation}
\label{eq:mass-moving-gap}
\int_0^t xf(x)\,\mathrm{d}x-\int_0^n xf_1(x)\,\mathrm{d}x
=\int_n^t xf(x)\,\mathrm{d}x-\frac{s}{n}\int_0^n x\,\mathrm{d}x
\ge \frac{ns}{2}\ge0.
\end{equation}

Set $f^*(x)=h-Lx$ on $[0,n]$. Since
$\int_0^n f^*(x)\,\mathrm{d}x=h^2/(2L)=m$, the function
$g=f_1-f^*$ has zero integral. Moreover, for $0\le x<y\le n$,
\[
g(y)-g(x)=f_1(y)-f_1(x)+L(y-x)\ge0,
\]
so $g$ is nondecreasing. Chebyshev's integral inequality therefore gives
$\int_0^n xg(x)\,\mathrm{d}x\ge0$. Combining this with
\eqref{eq:mass-moving-gap},
\[
\int_0^t xf(x)\,\mathrm{d}x
\ge \int_0^n x(h-Lx)\,\mathrm{d}x
=\frac{h^3}{6L^2}
=\frac{\sqrt2}{3\sqrt L}m^{3/2}.
\]

For equality, \eqref{eq:mass-moving-gap} forces $s=0$. The Chebyshev
gap satisfies
\[
\int_0^n xg(x)\,\mathrm{d}x
=\frac{1}{2n}\int_0^n\!\int_0^n
(x-y)\bigl(g(x)-g(y)\bigr)\,\mathrm{d}x\,\mathrm{d}y.
\]
Its integrand is nonnegative, and equality forces the continuous,
nondecreasing function $g$ to be constant; since $\int_0^n g=0$, this
constant is zero. Hence $f=h-Lx$ on $[0,n]$ and $f=0$ on $(n,t]$.
Conversely, $f(x)=\max\{h-Lx,0\}$ is $L$-Lipschitz, has mass $m$, and
its first moment is $h^3/(6L^2)$, so it attains equality.
\end{proof}

\begin{proof}[Proof of~\Cref{thm:lipschitz_threshold}]
For $L\ge4/9$,
\[
T_\alpha^{(L)}\ge\sqrt{\frac{2(1-\alpha)}{L}},
\]
because this is equivalent to
$L\ge2(1-\alpha)(1+2\alpha)^2/9$, whose right-hand side is at most
$4/9$.

Fix $E\in\mathcal E_L$ with density $f$ and let
$\alpha'=\mathbb P(E\ge t)$. First take $t=T_\alpha^{(L)}$ and suppose,
for contradiction, that $\alpha'>\alpha$. Then
$t\ge\sqrt{2(1-\alpha')/L}$, so \Cref{lemma:lipschitz} applies on
$[0,t]$. With
$c_L=\sqrt2/(3\sqrt L)$,
\[
1\ge\mathbb E[E]
\ge c_L(1-\alpha')^{3/2}+t\alpha',
\]
where the two terms lower-bound the moments on $[0,t]$ and
$[t,\infty)$, respectively. Hence
\[
t\le G(\alpha'),\qquad
G(u)=\frac{1-c_L(1-u)^{3/2}}{u}.
\]
Moreover,
\[
G'(u)=\frac{-1+c_L\sqrt{1-u}(1+u/2)}{u^2}<0,
\]
since
$[\sqrt{1-u}(1+u/2)]^2=1-3u^2/4-u^3/4\le1$ and
$c_L\le1/\sqrt2<1$. Thus
\[
t\le G(\alpha')<G(\alpha)=T_\alpha^{(L)},
\]
which is impossible because $t=T_\alpha^{(L)}$. Therefore
$\mathbb P(E\ge T_\alpha^{(L)})\le\alpha$.

Finally, if $t>T_\alpha^{(L)}$ and $\alpha'\ge\alpha$, the same moment
argument gives $t\le G(\alpha')\le G(\alpha)=T_\alpha^{(L)}$, a
contradiction. Hence $\mathbb P(E\ge t)<\alpha$ for every
$t>T_\alpha^{(L)}$.
\end{proof}

\subsection{Proof of \texorpdfstring{\Cref{thm:lip_property}}{Theorem \getrefnumber{thm:lip_property}}}
\label{app:sharpness-proof}
\begin{proof}
Write
\begin{equation*}
    T_u^{(L)}
    =
    \frac{1}{u}
    \left(
        1-\frac{\sqrt2}{3\sqrt L}(1-u)^{3/2}
    \right),
    \qquad
    t_m(u,L)
    =
    T_u^{(L)}-\sqrt{\frac{u}{L}} .
\end{equation*}
{Define $\kappa_L \triangleq 1-\frac{\sqrt2}{3\sqrt L}>0$ because $L>4/9$, and set}
\begin{equation*}
    {u_0(L)\triangleq \frac{\kappa_L\sqrt L}{1+\sqrt2}.}
\end{equation*}
Fix an alternative threshold satisfying
\begin{equation*}
    T_{\mathrm{new}}\le T_\alpha^{(L)}-\varepsilon_L .
\end{equation*}
Since $\sqrt{\alpha/L}\to0$ as $\alpha\to0$, for all sufficiently
small $\alpha$ we have
\begin{equation*}
    t_m(\alpha,L)
    =
    T_\alpha^{(L)}-\sqrt{\frac{\alpha}{L}}
    \ge
    T_\alpha^{(L)}-\frac{\varepsilon_L}{2}.
\end{equation*}
Therefore
\begin{equation*}
    T_{\mathrm{new}}
    \le
    T_\alpha^{(L)}-\varepsilon_L
    \le
    t_m(\alpha,L)-\frac{\varepsilon_L}{2}.
\end{equation*}
Since $t_m(u,L)$ is continuous in $u$, we may choose
$\alpha_1>\alpha$ sufficiently close to $\alpha$ such that
\begin{equation*}
    t_m(\alpha_1,L)\ge T_{\mathrm{new}} .
\end{equation*}
Moreover, because the theorem assumes $\alpha$ is sufficiently small, we can
{choose $\alpha_1$ additionally satisfying $\alpha_1<\min\{1,u_0(L)\}$. Then}
\begin{equation*}
    t_m(\alpha_1,L)>\sqrt{\frac{2(1-\alpha_1)}{L}},
\end{equation*}
which ensures the head and tail supports are disjoint.
We now construct the separated extremal distribution with tail mass
$\alpha_1$.

Define the head density by
\begin{equation*}
    f_{\mathrm{head}}(x)
    =
    \max\{\sqrt{2L(1-\alpha_1)}-Lx,0\}.
\end{equation*}
It has mass $1-\alpha_1$ and first moment
\begin{equation*}
    \int_0^\infty x f_{\mathrm{head}}(x)\,\mathrm{d}x
    =
    \frac{\sqrt2}{3\sqrt L}(1-\alpha_1)^{3/2}.
\end{equation*}
For the tail, set $a_{\alpha_1}=\sqrt{\alpha_1/L}$ and define
\begin{equation*}
    f_{\mathrm{tail}}(x)
    =
    \begin{cases}
        L(x-t_m(\alpha_1,L)),
        & t_m(\alpha_1,L)\le x\le t_m(\alpha_1,L)+a_{\alpha_1},\\
        L(t_m(\alpha_1,L)+2a_{\alpha_1}-x),
        & t_m(\alpha_1,L)+a_{\alpha_1}<x\le t_m(\alpha_1,L)+2a_{\alpha_1},\\
        0, & \text{otherwise}.
    \end{cases}
\end{equation*}
The head support is
\begin{equation*}
    \mathrm{supp}(f_{\mathrm{head}})
    =
    \left[0,\sqrt{\frac{2(1-\alpha_1)}{L}}\right],
\end{equation*}
and the tail starts at $t_m(\alpha_1,L)$. We now verify separation.
For any $u\in(0,1)$, since $(1-u)^{3/2}\le 1$, we have
\begin{equation*}
    {T_u^{(L)}
    \ge
    \frac{1}{u}\left(1-\frac{\sqrt2}{3\sqrt L}\right)
    =
    \frac{\kappa_L}{u},}
\end{equation*}
hence
\begin{equation*}
    {t_m(u,L)
    =
    T_u^{(L)}-\sqrt{\frac{u}{L}}
    \ge
    \frac{\kappa_L}{u}-\sqrt{\frac{u}{L}}
    \ge
    \frac{\kappa_L}{u}-\frac{1}{\sqrt L}.}
\end{equation*}
Therefore, if $u<u_0(L)$, then
\begin{equation*}
    t_m(u,L)
    >
    \frac{1+\sqrt2}{\sqrt L}-\frac{1}{\sqrt L}
    =
    \sqrt{\frac{2}{L}}
    \ge
    \sqrt{\frac{2(1-u)}{L}}.
\end{equation*}
Because $\alpha$ is sufficiently small, we can choose $\alpha_1>\alpha$
close enough to $\alpha$ while still satisfying $\alpha_1<u_0(L)$.
Thus
\begin{equation*}
    t_m(\alpha_1,L)>\sqrt{\frac{2(1-\alpha_1)}{L}},
\end{equation*}
and hence
\begin{equation*}
    \mathrm{supp}(f_{\mathrm{head}})
    \cap
    \mathrm{supp}(f_{\mathrm{tail}})
    =
    \varnothing .
\end{equation*}

The tail has mass $\alpha_1$, left endpoint $t_m(\alpha_1,L)$, and first
moment
\begin{equation*}
    \int_0^\infty x f_{\mathrm{tail}}(x)\,\mathrm{d}x
    =
    \alpha_1 t_m(\alpha_1,L)
    +
    \frac{\alpha_1^{3/2}}{\sqrt L}.
\end{equation*}
Hence $f^\ast=f_{\mathrm{head}}+f_{\mathrm{tail}}$ has total mass one and
\begin{equation*}
\begin{aligned}
    \mathbb E[E]
    &=
    \frac{\sqrt2}{3\sqrt L}(1-\alpha_1)^{3/2}
    +
    \alpha_1 t_m(\alpha_1,L)
    +
    \frac{\alpha_1^{3/2}}{\sqrt L}  \\
    &=
    \frac{\sqrt2}{3\sqrt L}(1-\alpha_1)^{3/2}
    +
    \alpha_1
    \left(
        T_{\alpha_1}^{(L)}-\sqrt{\frac{\alpha_1}{L}}
    \right)
    +
    \frac{\alpha_1^{3/2}}{\sqrt L}
    =
    1.
\end{aligned}
\end{equation*}
Thus $f^\ast$ defines a valid separated $L$-Lipschitz e-variable
construction. Moreover, $f^\ast$ is globally $L$-Lipschitz: it is a
continuous piecewise-affine function with segment slopes in
$\{-L,0,+L\}$, hence $|(f^\ast)'(x)|\le L$ almost everywhere.
Therefore, for any $x<y$,
\begin{equation*}
|f^\ast(y)-f^\ast(x)|
=
\left|\int_x^y (f^\ast)'(u)\,\mathrm{d}u\right|
\le
\int_x^y |(f^\ast)'(u)|\,\mathrm{d}u
\le
L|y-x|.
\end{equation*}

Finally, since $T_{\mathrm{new}}\le t_m(\alpha_1,L)$, the entire tail
mass lies above $T_{\mathrm{new}}$. Therefore
\begin{equation*}
    \Pr(E\ge T_{\mathrm{new}})
    \ge
    \alpha_1
    >
    \alpha .
\end{equation*}
Hence every threshold satisfying
\begin{equation*}
    T_{\mathrm{new}}\le T_\alpha^{(L)}-\varepsilon_L
\end{equation*}
is invalid for level $\alpha$, because a valid $L$-Lipschitz
e-variable can be constructed whose rejection probability exceeds
$\alpha$. Consequently, $T_\alpha^{(L)}$ cannot be lowered by any fixed
$L$-dependent additive margin while maintaining uniform
level-$\alpha$ validity in the small-$\alpha$ regime.
\end{proof}

\subsection{Proof of \texorpdfstring{\Cref{prop:lipschitz_fdr}}{Proposition \getrefnumber{prop:lipschitz_fdr}}}

\begin{proof}
    We reduce the e-LC step-up procedure to the classical BH procedure applied to induced p-values. Throughout the proof, the independence assumption is the stated mutual independence of the e-values $(E_1,\ldots,E_K)$.
    {For $u\in(0,1)$, \Cref{thm:lipschitz_threshold} directly gives}
    \begin{equation*}
        {\mathbb P_0(E_i\ge T_u^{(L)})\le u,\qquad i\in\mathcal H_0.}
    \end{equation*}
    {Extend the formula by $T_1^{(L)}=1$, for which the bound by $1$ is trivial. For each hypothesis, define the induced p-value}
    \begin{equation*}
        \tilde p_i = \inf\left(\left\{u\in(0,1]:E_i \ge T_u^{(L)}\right\}\cup\{1\}\right).
    \end{equation*}
    We first show that each true-null $\tilde p_i$ is super-uniform. Fix $i\in\mathcal H_0$ and $x\in(0,1)$. For any $\varepsilon>0$ such that $x+\varepsilon<1$, we claim that
    \begin{equation*}
        \{\tilde p_i\le x\}\subseteq\{E_i\ge T_{x+\varepsilon}^{(L)}\}.
    \end{equation*}
    On the event $\{\tilde p_i\le x\}$, the set $\left\{u\in(0,1]:E_i \ge T_u^{(L)}\right\}\cup\{1\}$ is nonempty and has infimum at most $x$. Therefore there exists $v$ in this set such that $v\le x+\varepsilon$. Since $x+\varepsilon<1$, this $v$ cannot be the added singleton $1$. Hence $v\in(0,1]$ and $E_i\ge T_v^{(L)}$. By monotonicity of $u\mapsto T_u^{(L)}$, $v\le x+\varepsilon$ implies
    \begin{equation*}
        T_v^{(L)}\ge T_{x+\varepsilon}^{(L)}.
    \end{equation*}
    Thus $E_i\ge T_{x+\varepsilon}^{(L)}$, proving the inclusion. By exact-threshold validity,
    \begin{equation*}
        \mathbb P_0(\tilde p_i\le x) \le \mathbb P_0(E_i\ge T_{x+\varepsilon}^{(L)}) \le x+\varepsilon.
    \end{equation*}
    Letting $\varepsilon\to 0^+$ gives
    \begin{equation*}
        \mathbb P_0(\tilde p_i\le x)\le x, \quad x\in(0,1).
    \end{equation*}
    Since $\mathbb P_0(\tilde p_i\le1)=1$, $\tilde p_i$ is super-uniform.
    Next, we prove the exact threshold equivalence. Fix $a\in(0,1)$. We claim that, for every $i$,
    \begin{equation*}
        E_i\ge T_a^{(L)} \Longleftrightarrow \tilde p_i\le a.
    \end{equation*}
    If $E_i\ge T_a^{(L)}$, then $a$ belongs to $\{u\in(0,1]:E_i\ge T_u^{(L)}\}$, and therefore $\tilde p_i\le a$. Conversely, suppose $\tilde p_i\le a$. For every $\varepsilon>0$ such that $a+\varepsilon<1$, the definition of the infimum gives some $v\in\left\{u\in(0,1]:E_i \ge T_u^{(L)}\right\}\cup\{1\}$ satisfying $v\le a+\varepsilon$. Since $a+\varepsilon<1$, this $v$ cannot be $1$. Hence $v\in(0,1]$ and $E_i\ge T_v^{(L)}$. By monotonicity,
    \begin{equation*}
        E_i\ge T_v^{(L)} \ge T_{a+\varepsilon}^{(L)}.
    \end{equation*}
    Therefore $E_i\ge T_{a+\varepsilon}^{(L)}$ for every sufficiently small $\varepsilon>0$. Since $u\mapsto T_u^{(L)}$ is right-continuous at $a$, letting $\varepsilon\to 0^+$ yields
    \begin{equation*}
        E_i\ge T_a^{(L)}.
    \end{equation*}
    This proves the equivalence.
    Let $E_{(1)}\ge\cdots\ge E_{(K)}$ denote the ordered e-values, and let $\tilde p_{(1)}\le\cdots\le\tilde p_{(K)}$ denote the ordered induced p-values. The threshold equivalence implies that, for every $a\in(0,1)$,
    \begin{equation*}
        \#\{i:E_i\ge T_a^{(L)}\} = \#\{i:\tilde p_i\le a\}.
    \end{equation*}
    Hence, for every $k\in\{1,\ldots,K\}$, since $k \alpha / K \in (0, 1)$,
    \begin{equation*}
        E_{(k)} \ge T_{k \alpha/K}^{(L)} \Longleftrightarrow \tilde p_{(k)}\le \frac{k\alpha}{K}.
    \end{equation*}
    Therefore the e-LC step-up index satisfies
    \begin{equation*}
        {R_{\mathrm{LC}} = \max\left(\{0\}\cup\left\{k \in \mathcal K :\tilde p_{(k)}\le \frac{k\alpha}{K}\right\}\right),}
    \end{equation*}
    which is exactly the BH step-up index applied to $\tilde p_1,\ldots,\tilde p_K$.
    We use the rejection convention
    \begin{equation*}
        \mathcal R_{\mathrm{LC}} = \left\{i:E_i\ge T_{R_{\mathrm{LC}}\alpha/K}^{(L)}\right\}
    \end{equation*}
    when $R_{\mathrm{LC}}>0$, and $\mathcal R_{\mathrm{LC}}=\emptyset$ when $R_{\mathrm{LC}}=0$. If $R_{\mathrm{LC}}>0$, the threshold equivalence with $a=R_{\mathrm{LC}}\alpha/K$ gives
    \begin{equation*}
        \mathcal R_{\mathrm{LC}} = \left\{i:\tilde p_i\le \frac{R_{\mathrm{LC}}\alpha}{K}\right\}.
    \end{equation*}
    Thus e-LC has exactly the same rejection set as BH applied to the induced p-values.
    It remains to verify the assumptions needed for the BH theorem under independence. We have already shown that every true-null induced p-value is super-uniform. Moreover, the map $x\mapsto\inf\left(\left\{u\in(0,1]:x \ge T_u^{(L)}\right\}\cup\{1\}\right)$ is measurable, so each $\tilde p_i$ is a measurable function of $E_i$. Since $(E_1,\ldots,E_K)$ are mutually independent by assumption, the induced p-values $(\tilde p_1,\ldots,\tilde p_K)$ are also mutually independent. Therefore, the classical BH theorem under independence applies and gives
    \begin{equation*}
        \mathrm{FDR}(\mathcal R_{\mathrm{LC}}) \le \frac{K_0}{K}\alpha.
    \end{equation*}
    Since $K_0\le K$, we also have $\mathrm{FDR}(\mathcal R_{\mathrm{LC}})\le\alpha$.
    Finally, the threshold comparison follows from $T_u^{(L)}\le 1/u$ for all $u\in(0,1)$. Substituting $u=k\alpha/K$ gives
    \begin{equation*}
        T_{k\alpha/K}^{(L)} \le \frac{K}{k\alpha}, \quad k \in \mathcal K.
    \end{equation*}
    Hence the e-LC threshold sequence is uniformly no larger than the e-BH threshold sequence.
\end{proof}

\subsection{Proof of \texorpdfstring{\Cref{thm:lipschitz-truncated-ebh}}{Theorem \getrefnumber{thm:lipschitz-truncated-ebh}}}
\label{app:lip-truncated-ebh-proof}

\begin{proof}[Proof of \Cref{thm:lipschitz-truncated-ebh}]
    Recall that
    \begin{equation*}
        C_L=\frac{\sqrt{2}}{3\sqrt L}, \quad d_{\alpha,L}=1-C_L(1-\alpha)^{3/2}, \quad T_\alpha^{(L)}=\frac{d_{\alpha,L}}{\alpha},
    \end{equation*}
    and
    \begin{equation*}
        b_{\alpha,L}(e)=\frac{e}{d_{\alpha,L}}\mathbf{1}\{e\ge T_\alpha^{(L)}\}, \quad E_i^{\mathfrak b}=b_{\alpha,L}(E_i).
    \end{equation*}
    Fix a true null and write $E$ for its original e-value, $E^{\mathfrak b}=b_{\alpha,L}(E)$ for its boosted e-value, and $f$ for its null density. Let
    \begin{equation*}
        t=T_\alpha^{(L)}.
    \end{equation*}
    By \Cref{thm:lipschitz_threshold},
    \begin{equation*}
        \mathbb P_0(E\ge t)\le\alpha.
    \end{equation*}
    Define the head mass
    \begin{equation*}
        m=\int_0^t f(x)dx.
    \end{equation*}
    Then $m=\mathbb P_0(E<t)\ge1-\alpha$. From the proof of \Cref{thm:lipschitz_threshold}, the condition $L\ge4/9$ implies
    \begin{equation*}
        T_\alpha^{(L)}\ge \sqrt{\frac{2(1-\alpha)}{L}}.
    \end{equation*}
    Since $t=T_\alpha^{(L)}$, this is equivalent to
    \begin{equation*}
        1-\alpha\le \frac{Lt^2}{2}.
    \end{equation*}
    Set $c=(1-\alpha)/m\le1$ and $g=cf$ on $[0,t]$. Then $g$ is nonnegative and $L$-Lipschitz, and it has mass $1-\alpha$. The preceding display is precisely the condition needed to apply \Cref{lemma:lipschitz} to $g$. Since $g\le f$,
    \begin{equation*}
        \int_0^t xf(x)dx
        \ge \int_0^t xg(x)dx
        \ge C_L(1-\alpha)^{3/2}.
    \end{equation*}
    Using the e-value condition $\mathbb E_0[E]\le1$,
    \begin{equation*}
        \mathbb E_0[E\mathbf{1}\{E\ge t\}]=\mathbb E_0[E]-\int_0^t xf(x)dx\le 1-C_L(1-\alpha)^{3/2}=d_{\alpha,L}.
    \end{equation*}
    Thus
    \begin{equation*}
        \mathbb E_0[E^{\mathfrak b}]=\mathbb E_0\left[\frac{E}{d_{\alpha,L}}\mathbf{1}\{E\ge T_\alpha^{(L)}\}\right]\le1.
    \end{equation*}
    Hence $E^{\mathfrak b}$ is a valid e-value under the true null.

    Now apply standard e-BH at level $\alpha$ to $E_1^{\mathfrak b},\ldots,E_K^{\mathfrak b}$. Let
    \begin{equation*}
        R=\max\left\{k\in\{1,\ldots,K\}:E_{(k)}^{\mathfrak b}\ge \frac{K}{\alpha k}\right\},
    \end{equation*}
    with no rejections if the set is empty. Let $\mathcal R$ be the rejected set and let $V=|\mathcal R\cap\mathcal H_0|$. If $R=0$, then $V/(R\vee1)=0$. If $R>0$ and $i\in\mathcal R\cap\mathcal H_0$, then
    \begin{equation*}
        E_i^{\mathfrak b}\ge \frac{K}{\alpha R}.
    \end{equation*}
    Hence
    \begin{equation*}
        \frac{\mathbf{1}\{i\in\mathcal R\}}{R\vee1}\le \frac{\alpha E_i^{\mathfrak b}}{K}.
    \end{equation*}
    Summing over true nulls gives
    \begin{equation*}
        \frac{V}{R\vee1}\le \sum_{i\in\mathcal H_0}\frac{\alpha E_i^{\mathfrak b}}{K}.
    \end{equation*}
    Taking expectations,
    \begin{equation*}
        \mathrm{FDR}=\mathbb E\left[\frac{V}{R\vee1}\right]\le \frac{\alpha}{K}\sum_{i\in\mathcal H_0}\mathbb E_0[E_i^{\mathfrak b}]\le \frac{K_0}{K}\alpha\le\alpha.
    \end{equation*}
    This argument does not use independence or positive dependence, so the result holds under arbitrary dependence.

    Finally, since $b_{\alpha,L}$ is nondecreasing, $(E^{\mathfrak b})_{(k)}=b_{\alpha,L}(E_{(k)})$. The e-BH condition
    \begin{equation*}
        (E^{\mathfrak b})_{(k)}\ge \frac{K}{\alpha k}
    \end{equation*}
    is equivalent to
    \begin{equation*}
        E_{(k)}\ge \frac{d_{\alpha,L}K}{\alpha k},
    \end{equation*}
    because
    \begin{equation*}
        \frac{d_{\alpha,L}K}{\alpha k}\ge \frac{d_{\alpha,L}}{\alpha}=T_\alpha^{(L)}
    \end{equation*}
    for every $k\le K$. Hence the original-scale thresholds are
    \begin{equation*}
        t_k^{\mathfrak b}=\frac{d_{\alpha,L}K}{\alpha k}, \quad k \in \mathcal K.
    \end{equation*}
    Since $0<d_{\alpha,L}<1$ under $L\ge4/9$,
    \begin{equation*}
        t_k^{\mathfrak b}<\frac{K}{\alpha k}.
    \end{equation*}
    Moreover, to compare with the original e-LC threshold, write
    \begin{equation*}
        T_{k\alpha/K}^{(L)}
        =
        \frac{K}{k\alpha}
        \left(
            1-C_L\left(1-\frac{k\alpha}{K}\right)^{3/2}
        \right),
        \qquad
        t_k^{\mathfrak b}
        =
        \frac{K}{k\alpha}
        \left(
            1-C_L(1-\alpha)^{3/2}
        \right).
    \end{equation*}
    Since $k\alpha/K\le\alpha$, we have
    $1-k\alpha/K\ge 1-\alpha$, hence
    \begin{equation*}
        \left(1-\frac{k\alpha}{K}\right)^{3/2}\ge (1-\alpha)^{3/2},
    \end{equation*}
    and therefore
    \begin{equation*}
        1-C_L\left(1-\frac{k\alpha}{K}\right)^{3/2}
        \le
        1-C_L(1-\alpha)^{3/2}.
    \end{equation*}
    Multiplying by $K/(k\alpha)>0$ yields
    \begin{equation*}
        T_{k\alpha/K}^{(L)}\le t_k^{\mathfrak b}.
    \end{equation*}
    Therefore
    \begin{equation*}
        T_{k\alpha/K}^{(L)}\le t_k^{\mathfrak b}\le \frac{K}{\alpha k}, \quad k \in \mathcal K.
    \end{equation*}
\end{proof}

\subsection{Proof of \texorpdfstring{\Cref{lemma:supermartingale}}{Lemma \getrefnumber{lemma:supermartingale}}}
\label{app:lemma-supermartingale-proof}
\begin{proof}[Proof of \Cref{lemma:supermartingale}]
    Let
    \begin{equation*}
        M_0=1,\qquad M_t=\prod_{s=1}^t \phi_s(E_s),\quad t\ge 1.
    \end{equation*}
    {Since each $\phi_s$ is nonnegative, we have $M_t\ge 0$ almost surely. Let $\mathcal D_L$ denote the densities of variables in $\mathcal E_L$, and let $\mathcal D_L^0$ denote those with $f(0)=0$.}
    It remains to verify the supermartingale inequality.

    Fix $t\ge 1$. By assumption, conditionally on $\mathcal F_{t-1}$, the
    law of $E_t$ belongs to $\mathcal E_L$ and satisfies
    $\mathbb E[E_t\mid\mathcal F_{t-1}]\le 1$. Therefore, by definition of
    \begin{equation*}
        {\Psi_L(\phi_t)
        =
        \sup_{f\in\mathcal D_L,\,\mathbb E_f[X]\le 1}
        \int_0^\infty \phi_t(x)f(x)\,\mathrm dx,}
    \end{equation*}
    the condition $\Psi_L(\phi_t)\le 1$ implies
    \begin{equation*}
        \mathbb E[\phi_t(E_t)\mid\mathcal F_{t-1}] \le 1
    \end{equation*}
    almost surely.

    Using the multiplicative update $M_t=M_{t-1}\phi_t(E_t)$ and the fact
    that $M_{t-1}$ is $\mathcal F_{t-1}$-measurable,
    \begin{equation*}
        \mathbb E[M_t\mid\mathcal F_{t-1}]
        =
        \mathbb E[M_{t-1}\phi_t(E_t)\mid\mathcal F_{t-1}]
        =
        M_{t-1}\,\mathbb E[\phi_t(E_t)\mid\mathcal F_{t-1}]
        \le
        M_{t-1}.
    \end{equation*}
    {Induction gives $\mathbb E[M_t]\le1$, so each $M_t$ is integrable. Hence $(M_t)_{t\ge 0}$ is a nonnegative supermartingale with respect to}
    $(\mathcal F_t)$.
\end{proof}

\subsection{Proof of \texorpdfstring{\Cref{thm:process-based}}{Theorem \getrefnumber{thm:process-based}}}
\label{app:process-based-proof}
\begin{proof}[Proof of~\Cref{thm:process-based}]
We show that the linear update family
$\phi_{\lambda,\eta,\tau}(x) \triangleq (1-\lambda)+\lambda \eta (x-\tau)_+$
admits an explicit structural calibration under the boundary-restricted
Lipschitz class $\mathcal E_L^0$, and that the resulting multiplicative
process is anytime valid.
We first derive an explicit lower bound on $L$.
{Fix any $f\in\mathcal D_L^0$ with $\mathbb E_f[X]\le 1$.}
Since $f(0)=0$ and $f$ is $L$-Lipschitz, we have
$f(x)\le Lx$ for all $x\ge 0$. Hence, for
$F(x)\triangleq \int_0^x f(u)\,\mathrm du$,
\begin{equation*}
    F(x)\le \int_0^x Lu\,\mathrm du=\frac{Lx^2}{2}.
\end{equation*}
Using the tail-integral identity for nonnegative $X$,
\begin{equation*}
\begin{aligned}
    \mathbb E_f[X]
    &=\int_0^\infty \!\bigl(1-F(x)\bigr) \mathrm dx \ge \int_0^\infty \!\left(1-\min\left\{1,\frac{Lx^2}{2}\right\}\right)\mathrm{d}x \\
    &=\int_0^{\sqrt{2/L}} \!\left(1-\frac{Lx^2}{2}\right)\mathrm{d}x
    =\frac{2\sqrt2}{3\sqrt L}.
\end{aligned}
\end{equation*}
Since $\mathbb E_f[X]\le 1$, we must have $\frac{2\sqrt2}{3\sqrt L}\le 1$, so $L\ge \frac{8}{9}$.

Fix $\lambda \in [0,1]$, $\eta \ge 0$, and
$\tau \in [0,1/\sqrt{2L}]$. By definition of the worst-case expectation
operator $\Psi_L^0$,
\begin{equation*}
    {\Psi_L^0(\phi_{\lambda,\eta,\tau})
    =
    \sup_{f \in \mathcal D_L^0,\, \mathbb E_f[X]\le 1}
    \int_0^\infty
    \bigl((1-\lambda)+\lambda \eta (x-\tau)_+\bigr)
    f(x)\,\mathrm{d}x .}
\end{equation*}
Using $\int_0^\infty f(x)\mathrm{d}x=1$ for any density $f$, we obtain
\begin{equation*}
    {\Psi_L^0(\phi_{\lambda,\eta,\tau})
    =
    (1-\lambda)+\lambda\eta
    \sup_{f \in \mathcal D_L^0,\, \mathbb E_f[X]\le 1}
    \mathbb E_f[(X-\tau)_+] .}
\end{equation*}
Therefore, it suffices to upper bound
\begin{equation*}
    {H_L^0(\tau)
    \triangleq
    \sup_{f \in \mathcal D_L^0,\, \mathbb E_f[X]\le 1}
    \mathbb E_f[(X-\tau)_+] .}
\end{equation*}

{For $\tau=0$, the desired positive-part bound is simply $\mathbb E_f[X]\le1$. Hence assume $\tau>0$.}
{Fix any $f\in\mathcal D_L^0$ with $\mathbb E_f[X]\le 1$, and set}
\begin{equation*}
    q\triangleq \int_0^\tau f(x)\,\mathrm{d}x
    =
    \mathbb P_f(X<\tau).
\end{equation*}
Since $f(0)=0$ and $f$ is $L$-Lipschitz, we have $f(x)\le Lx$ for
all $x\ge0$. Hence
\begin{equation*}
    q
    =
    \int_0^\tau f(x)\,\mathrm{d}x
    \le
    \int_0^\tau Lx\,\mathrm{d}x
    =
    \frac{L\tau^2}{2}.
\end{equation*}
Equivalently, $\tau\ge \sqrt{2q/L}$. Therefore the condition of
\Cref{lemma:lipschitz} is satisfied on the interval $[0,\tau]$ with
head mass $q$, and we get
\begin{equation*}
    \int_0^\tau x f(x)\,\mathrm{d}x
    \ge
    \frac{\sqrt2}{3\sqrt L}q^{3/2}.
\end{equation*}

Now decompose the positive-part expectation as
\begin{equation*}
    \mathbb E_f[(X-\tau)_+]
    =
    \mathbb E_f[X]-\tau+\int_0^\tau(\tau-x)f(x)\,\mathrm{d}x .
\end{equation*}
Since $\mathbb E_f[X]\le 1$, this gives
\begin{equation*}
    \mathbb E_f[(X-\tau)_+]
    \le
    1-\tau+\tau q-\int_0^\tau x f(x)\,\mathrm{d}x .
\end{equation*}
Applying the preceding lower bound,
\begin{equation*}
    \mathbb E_f[(X-\tau)_+]
    \le
    1-\tau+\tau q-\frac{\sqrt2}{3\sqrt L}q^{3/2}.
\end{equation*}
On the interval $0\le q\le L\tau^2/2$, the function
$\tau q-\frac{\sqrt2}{3\sqrt L}q^{3/2}$ is increasing, since
its derivative is
$\tau-\frac{\sqrt2}{2\sqrt L}\sqrt q\ge \tau/2\ge0$.
Therefore its maximum over this interval is attained at
$q=L\tau^2/2$, and hence
\begin{equation*}
    \tau q-\frac{\sqrt2}{3\sqrt L}q^{3/2}
    \le
    \frac{1}{3}L\tau^3.
\end{equation*}
Consequently,
\begin{equation*}
    H_L^0(\tau)
    \le
    1-\tau+\frac{1}{3}L\tau^3
    \eqqcolon U_L(\tau).
\end{equation*}
We now verify that $U_L(\tau)>0$ on the admissible range
$\tau\in[0,1/\sqrt{2L}]$. Since
\begin{equation*}
    U_L'(\tau)=-1+L\tau^2\le -1+\frac12=-\frac12<0
\end{equation*}
on this interval, $U_L$ is decreasing in $\tau$, so its minimum is at
$\tau=1/\sqrt{2L}$. Therefore
\begin{equation*}
    U_L(\tau)\ge U_L\!\left(\frac1{\sqrt{2L}}\right)
    =
    1-\frac{5}{6\sqrt{2L}}.
\end{equation*}
Since $L\ge 8/9$, we have $\sqrt{2L}\ge \sqrt{16/9}=4/3$, hence
\begin{equation*}
    U_L(\tau)\ge 1-\frac{5}{6\sqrt{2L}}
    \ge
    1-\frac{5}{8}>0.
\end{equation*}
Thus $U_L(\tau)$ is strictly positive on the admissible range, and
$\eta_L(\tau)=1/U_L(\tau)$ is well-defined and nonnegative.

Returning to $\Psi_L^0(\phi_{\lambda,\eta,\tau})$, we obtain
\begin{equation*}
    \Psi_L^0(\phi_{\lambda,\eta,\tau})
    \le
    (1-\lambda)+\lambda\eta U_L(\tau).
\end{equation*}
Thus the calibration
\begin{equation*}
    \eta_L(\tau)
    =
    \frac{1}{U_L(\tau)}
    =
    \frac{1}{1-\tau+\frac{1}{3}L\tau^3}
\end{equation*}
ensures
\begin{equation*}
    \Psi_L^0(\phi_{\lambda,\eta_L(\tau),\tau})
    \le
    (1-\lambda)+\lambda
    =
    1
\end{equation*}
for all $\lambda\in[0,1]$.

Now let $\phi_t=\phi_{\lambda_t,\eta_L(\tau_t),\tau_t}$ be any
predictable sequence with $\lambda_t\in[0,1]$ and
$0\le \tau_t\le 1/\sqrt{2L}$. Under the null, conditional on
$\mathcal F_{t-1}$, suppose the e-value $E_t$ has conditional density
{in $\mathcal D_L^0$ with conditional mean at most $1$. Then}
\begin{equation*}
    \mathbb E[\phi_t(E_t)\mid \mathcal F_{t-1}] \le 1.
\end{equation*}
Defining $M_0\triangleq 1$ and
$M_t\triangleq \prod_{s=1}^t \phi_s(E_s)$, we have
$\mathbb E[M_t\mid\mathcal F_{t-1}]\le M_{t-1}$, hence
$(M_t)_{t\ge0}$ is a nonnegative supermartingale. Ville's inequality
yields, for any $\alpha\in(0,1)$,
\begin{equation*}
    \mathbb P\Bigl(\sup_{t\ge1}M_t\ge 1/\alpha\Bigr)\le\alpha.
\end{equation*}
This establishes the claimed anytime validity.
\end{proof}

\section{Full Setup for the Main Text Experiments}
\label{app:exp_setup}

This section provides the complete implementation details for the numerical experiments in \Cref{sec:experiments}, complementing the concise descriptions in the main text. The setup is organized into two parts corresponding to the fixed-sample multiple-testing simulation (\Cref{subsec:exp_multiple}) and the sequential real-data experiment (\Cref{subsec:exp_sequential}).

\subsection{Multiple-Testing Simulation Setup}

The simulation aims to examine the power performance of the Lipschitz calibration in fixed-sample multiple testing. All procedures are applied to the same set of e-values, differing only in their threshold calibration. Competing methods include the standard Markov-calibrated e-BH~\citep{wang2022false} and the monotone benchmark~\citep{blier2024improved}.

\textbf{Data Generation.} Null e-values are generated from a normalized two-segment piecewise-linear profile. To distinguish density height from the null proportion, let $p_{\mathrm{ht}}$ denote the unnormalized height at zero and define
\begin{equation*}
\widetilde f(x)=p_{\mathrm{ht}}
\begin{cases}
1-k_1x, & 0\le x\le x_1,\\
1-k_1x_1-k_2(x-x_1), & x_1<x\le x_2,\\
0, & \text{otherwise},
\end{cases}
\end{equation*}
where $x_2=x_1+(1-k_1x_1)/k_2$ makes the two pieces continuous and brings the profile to zero. The null density is $f_0=\widetilde f/N$, where $N=\int\widetilde f$, and its Lipschitz constant is
\begin{equation*}
L=\frac{p_{\mathrm{ht}}}{N}\max\{k_1,k_2\}.
\end{equation*}
Candidate parameters satisfy $p_{\mathrm{ht}}>0$, $k_1,k_2,x_1>0$, and $k_1x_1<1$, which make the profile nonnegative and give $x_2>x_1$. After normalization, we retain only profiles with null mean $\mathbb E_0[E]\le1$ and $L\ge4/9$, as required for the stated e-LC calibration.
{This family is not the separated body-and-tail extremal construction used in \Cref{thm:lip_property}. Non-null e-values are sampled independently from $U(0,1000)$ to represent strong signals. Each repetition contains exactly $K=10{,}000$ hypotheses: $2{,}000$ nulls and $8{,}000$ alternatives.}

\textbf{Configurations and Calibration.} Among the feasible profiles, we selected Set 1 with $(k_1,k_2,x_1)=(0.63,0.10,1.53)$ for a visible e-LC gain over the monotonicity benchmark; Set 2 retains the valid profile $(0.62,0.44,0.20)$. This is an illustrative parameter choice, not a random sample of feasible profiles. The endpoint $x_2$ follows from the continuity formula above. Normalization gives null means approximately $0.706570$ and $0.728048$, and Lipschitz constants approximately $0.455446$ and $0.585116$. The raw height cancels on normalization. We evaluate the constants in exact rational arithmetic and round upward for calibration. These nonincreasing densities also permit the decreasing-density comparator with threshold $1/(2u)$ at rank level $u=k\alpha/K$~\citep{blier2024improved}; this comparator is distinct from the bounded nondecreasing class in Appendix~\ref{app:monotone-benchmark}.

\textbf{Evaluation Metric.} Each configuration uses 20 independent repetitions (PCG64 seed 42), with the same data used for all three methods and all 100 equally spaced levels $\alpha\in[0.02,0.10]$. Null draws use the analytic inverse CDF. We report mean rejection counts and population standard deviations across repetitions. Every method's threshold exceeds the null support throughout this grid, so all false-rejection counts are zero by construction; these panels measure discovery counts under strong alternatives, not stress-test FDR control.

\subsection{Sequential Real-Data Experiment Setup}


\begin{algorithm}[H]
\caption{LC-process procedure}
\label{alg:LC-process}
\KwIn{Conditional e-values $E_1,\ldots,E_T$ whose conditional null laws belong to $\mathcal E^0_{L_c}$, revealed sequentially; horizon $T\ge1$; level $\alpha\in(0,1)$; fixed $\lambda\in[0,1]$, $L_c\ge8/9$, and $\tau\in[0,1/\sqrt{2L_c}]$.}
\KwOut{$(D,\sigma)$: rejection indicator $D\in\{0,1\}$ and first crossing time $\sigma\in\{1,\ldots,T\}\cup\{\infty\}$; $\sigma=\infty$ means no crossing by $T$.}
$\eta \gets \bigl(1-\tau+L_c\tau^3/3\bigr)^{-1}$\;
$M \gets 1$\;
\For{$t=1$ \KwTo $T$}{
  Observe $E_t$\;
  $M \gets M\bigl[(1-\lambda)+\lambda\eta\max\{E_t-\tau,0\}\bigr]$\;
  \If{$M \ge 1/\alpha$}{
    \Return{$(1,t)$}\;
  }
}
\Return{$(0,\infty)$}\;
\end{algorithm}

This subsection provides the detailed setup for the two sequential experiments in \Cref{subsec:exp_sequential}. The Gaussian experiment is a controlled mechanistic demonstration of the suppression and amplification behavior of the LC update, while the S\&P 500 experiment provides a retrospective comparison on real financial paths. Both experiments start from $M_0=1$ and use the same Ville boundary $1/\alpha$. Algorithm~\ref{alg:LC-process} gives the theorem-backed LC procedure when the conditional assumptions of \Cref{thm:process-based} hold.

\textbf{Gaussian Mechanistic Design.} We construct a single Gaussian likelihood-ratio path
\[
E_t=\exp(\theta X_t-\theta^2/2),
\]
with $\theta=0.3$ and NumPy seed $20261007$. Under the null, $X_t\overset{\mathrm{iid}}{\sim}\mathcal{N}(0,1)$. For the mechanistic illustration, the observations are independent with
\[
X_t\sim\mathcal{N}\!\left(\frac{\log 0.025+\theta^2/2}{\theta},1\right),\qquad t\leq100,
\]
and
\[
X_t\sim\mathcal{N}\!\left(\frac{0.95+\theta^2/2}{\theta},1\right),\qquad 101\leq t\leq260.
\]
The same realized inputs are used by vanilla and by every LC calibration, so all methods share one weak-signal phase and one common strong-signal phase. Vanilla uses the one-step factor
\[
1-\lambda_t+\lambda_tE_t,
\]
while LC uses
\[
1-\lambda_t+\lambda_t\eta_L(E_t-\tau_L)_+,
\]
where
\[
\tau_L=\frac{1}{\sqrt{8L}},
\qquad
\eta_L=\left(1-\tau_L+\frac{L\tau_L^3}{3}\right)^{-1}.
\]
We use $\lambda_t=0.003$ before the phase switch and $\lambda_t=0.15$ afterward. No clipping or resetting is applied.

\textbf{Gaussian Null Model and Calibration.} Under the Gaussian null,
\[
\log E_t\sim\mathcal{N}(-\theta^2/2,\theta^2),
\qquad
\mathbb{E}_0[E_t]=1.
\]
The induced null distribution is lognormal and satisfies the boundary condition $f(0)=0$. Let
\[
L_0=\sup_{e>0}|f'(e)|=4.6535404
\]
denote its analytic Lipschitz constant. We compare
\[
L=cL_0,\qquad c\in\{0.5,1,2\}.
\]
For $c=1$ and $c=2$, $L$ is a valid upper bound on the null-density Lipschitz constant. For $c=0.5$, the working value underestimates the true Lipschitz constant, so the sufficient condition in \Cref{thm:process-based} does not apply directly. We therefore check the Gaussian null expectation
\[
\mathbb{E}_0\!\left[\eta_L(E_t-\tau_L)_+\right]\leq1
\]
directly for every reported calibration. On the displayed path, all weak-phase inputs satisfy $E_t<\tau_L$ for every LC calibration, so the truncated term is zero throughout the weak-signal phase. All strong-phase inputs satisfy $E_t>24/23$, so the calibrated LC term exceeds the vanilla term throughout the common strong-signal phase.

\textbf{Gaussian Evaluation.} All processes are evaluated on the same realized path and are compared against the Ville boundaries corresponding to $\alpha\in\{0.1,0.05\}$. We record the first crossing of $\log(1/\alpha)$ for each method. The displayed trajectory is designed to isolate the two LC mechanisms, with suppression in the weak-signal phase followed by amplification in the strong-signal phase. Parameters were tuned over $630$ configurations per panel for this mechanistic illustration. For the separate absolute-$L$ panel, we use $L\in\{2,1,0.5\}$ and set the post-switch value to $\lambda_t=0.2$, with the remaining construction unchanged.

We constructed a frozen panel of adjusted closing prices from Yahoo Finance
using \texttt{yfinance}. The ticker universe was taken from the S\&P 500
constituent list available when the panel was assembled. Its 503 ticker columns span January 3, 2017 to June 3, 2026. The universe is selected retrospectively and does not reconstruct historical index membership, so survivorship and selection effects may influence the results. The supplied panel contains no nonpositive prices. We omit \texttt{FDXF}, which has fewer than three selected prices, leaving 502 sequential stock paths. Positive observations alone do not establish either the conditional e-value mean constraint or the density boundary condition required by the sequential theory.    

\textbf{S\&P 500 Path Construction.} For each stock, we deterministically select the first and last available adjusted-close prices in each calendar year from 2017 through 2026. Of the 502 paths, 475 contain 20 selected prices and the remaining paths are shorter. After chronological ordering, consecutive selected prices define $E_{k,j}^{\mathrm{base}}=R_{k,j}=P_{k,t_j}/P_{k,t_{j-1}}$ for $j=1,\ldots,n_k$, where $n_k$ is the number of ratios for stock $k$. These intervals alternate between within-year first-to-last periods and short year-end transitions. A monitoring step therefore does not correspond to a fixed number of trading days.

\textbf{S\&P 500 Null and Update Rules.} The formal stock-level null is $\mathbb E[R_{k,j}\mid\mathcal F_{j-1}]\le1$ for every $j$, where $\mathcal F_{j-1}$ represents the available history before the current ratio is observed. The historical data do not provide verified null and alternative labels for individual stocks. We compare the vanilla factor $(1-\lambda)+\lambda E$ with the LC factor $(1-\lambda)+\lambda\eta_{L_{c,k}}(\tau_{c,k})(E-\tau_{c,k})_+$. Both processes start from one and reject at their first crossing of the common Ville boundary $1/\alpha$.

\textbf{S\&P 500 Working Lipschitz Scale.} We fit a lognormal density to each complete ratio path using the mean and sample standard deviation of its log ratios. Let $\hat L_k$ denote the maximum absolute finite-difference derivative of the fitted density over 1,000 equally spaced points on $[0.8\min_jR_{k,j},1.2\max_jR_{k,j}]$. We define the working scale as
\[
L_{c,k}=\max\{0.5,s\hat L_k\},\qquad s\in\{0.5,1,2\}.
\]
For each working scale, we use $\tau_{c,k}=0.5/\sqrt{2L_{c,k}}$ and $\eta_{L_{c,k}}(\tau_{c,k})=(1-\tau_{c,k}+L_{c,k}\tau_{c,k}^3/3)^{-1}$. The fitted grid estimate does not certify an upper bound on the true conditional null-density Lipschitz constant. Fitting the complete path also makes the working calibration depend on observations occurring after earlier updates. The empirical S\&P 500 procedure therefore does not establish the predictable-update or conditional-density assumptions of \Cref{thm:process-based}. The theorem-backed algorithm serves as a formal reference for the update construction and not as a validity certificate for these historical stock paths.

\textbf{S\&P 500 Evaluation.} We use 80 evenly spaced values of $\lambda$ in $[0.05,1]$, both $\alpha\in\{0.05,0.1\}$, and all three working scales. For each configuration, the rejection rate is the fraction of the 502 eligible stock paths whose cumulative process crosses $1/\alpha$ before the end of the available monitoring path. The same stock paths and ratio inputs are used for the vanilla and LC procedures at each value of $\lambda$ and $\alpha$. The reported curves therefore compare the frequency of Ville-boundary crossings under the different one-step update rules. Because the Lipschitz scales are retrospective plug-in quantities and the required conditional assumptions are not verified, these rejection rates are interpreted as an exploratory comparison of detection behavior and not as estimates of prospective statistical power or a validation of anytime-valid error control.

\subsubsection{{Terminal Randomization Comparison}}
\label{app:randomized-comparison}

\textbf{Reserved Returns and Monitoring Window.} Figure~\ref{fig:lc_ep_4_methods} compares ordinary and terminal randomized Ville calibration~\citep{ramdas2026randomized}. We use scale $s=1$ and restrict monitoring to $m_k=n_k-3$ steps, leaving three future ratios available for any rejection by that cutoff. Stocks \texttt{Q} and \texttt{SNDK} have only three ratios and are additionally excluded, leaving 500 paths. An early rejection still uses the three ratios immediately after its own crossing, which can precede the terminal cutoff. The LC path is truncated after fitting its calibration on all $n_k$ ratios, so the reserved future returns also enter the fit; this is not an independent calibration holdout.

Randomized Ville monitoring checks $1/\alpha$ strictly before the terminal step. If an early crossing occurs at $t_k$, its conditional rejection probability is $w_k=1$; otherwise set $t_k=m_k$ and $w_k=\min\{\alpha M_{m_k},1\}$. We integrate over the independent uniform randomization rather than simulate its draw. Expected rejection rate is $\sum_k w_k/500$, and randomized ASN is $\sum_k w_kt_k/\sum_k w_k$. For the corresponding post-rejection returns, we retain the probability-weighted summaries
\begin{equation*}
    {\bar r_w=\frac{\sum_k w_kr_k}{W},\qquad
    s_w^2=\frac{\sum_k w_k(r_k-\bar r_w)^2}{W-\sum_k w_k^2/W},
    \qquad W=\sum_k w_k,}
\end{equation*}
{and report $\bar r_w/s_w$ when at least two positive-weight returns and a positive variance denominator and standard deviation are available. These are ratios of weighted summaries, not expectations of the sample ASN or return/SD ratio over realized randomizations. The retained complete-path fitting and conditional-model limitations make both real-data figures exploratory retrospective comparisons, not demonstrations of prospective performance or anytime-valid error control.}

\begin{figure}[t]
\centering
\input{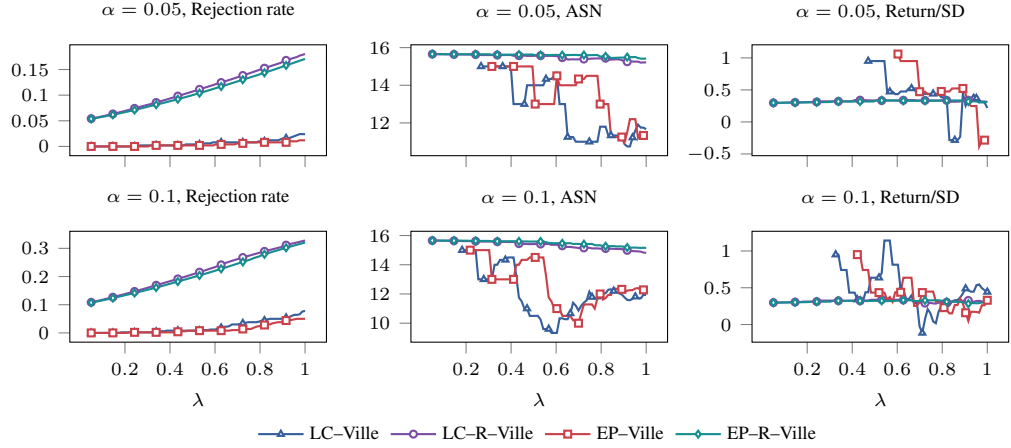}
\caption{{Retrospective comparison with terminal randomization at scale $s=1$, using 500 paths with three future ratios reserved. All four methods share the restricted monitoring window. Randomized metrics are probability-weighted summaries.}}
\label{fig:lc_ep_4_methods}
\end{figure}

In Figure~\ref{fig:lc_ep_4_methods}, randomization increases both processes' rejection rates; randomized LC has higher rates and smaller probability-weighted ASN than randomized EP over the grid, with mixed differences in return/SD. This illustrates their interaction in the retrospective calculation.

\section{Additional Experiments}

\subsection{{Computing the Lipschitz Constant under a Gaussian Model}}
\hypertarget{GaussianL}{}
In the case of
\begin{equation*}
    E = \frac{P(X \mid H_1)}{P(X \mid H_0)} = \exp\left(\mu X - \frac{\mu^2}{2}\right),
\end{equation*}
where $X \sim \mathcal{N}(0,1)$ under $H_0$, the random variable $E$ follows a log-normal distribution. 
For $\mu \neq 0$, we have $\log E \sim \mathcal{N}(m,s^2)$, where
\begin{equation*}
    m = -\frac{\mu^2}{2}, \qquad s = |\mu| > 0.
\end{equation*}
When $\mu=0$, $E=1$ almost surely, so its distribution is degenerate and does not admit a Lebesgue density.

Let $f_E(e)$ denote the probability density function of $E$. The corresponding $L$ value is defined as
\begin{equation*}
    L = \sup_{e > 0} \left| f_E'(e) \right|,
\end{equation*}
where $f_E'(e)$ denotes the derivative of $f_E(e)$ with respect to $e$. Note that \(e\) here denotes a variable, rather than Euler's number. 

The density of $E$ admits the closed-form expression
\begin{equation*}
    f_E(e) = \frac{1}{e s \sqrt{2\pi}} \exp\!\left(-\frac{(\log e - m)^2}{2 s^2}\right), \quad e > 0,
\end{equation*}
and its derivative is given by
\begin{equation*}
    f_E'(e) = f_E(e) \left( -\frac{\log e - m}{s^2 e} - \frac{1}{e} \right).
\end{equation*}

The Lipschitz constant can be computed from two analytic candidates. Set $z=\log e$ and write
\begin{equation*}
    {D(z)=\frac{|1+(z-m)/s^2|}{s\sqrt{2\pi}}
    \exp\!\left(-2z-\frac{(z-m)^2}{2s^2}\right).}
\end{equation*}
Away from the zero $z=m-s^2$, differentiating $\log D(z)$ gives
$-2-(z-m)/s^2+1/(s^2+z-m)=0$.
Its two solutions are
\begin{equation*}
    {z_\pm=m-\frac32s^2\pm\frac{s}{2}\sqrt{s^2+4},
    \qquad L=\max\{D(z_-),D(z_+)\}.}
\end{equation*}
Both endpoint limits of $D$ are zero, so these candidates give the global maximum. The density and its derivative also vanish at the origin, establishing the boundary condition for the sequential calibration when $s>0$.

For a fully specified Gaussian null and likelihood ratio, this expression gives a deterministic model-based Lipschitz constant; numerical implementations use an upward-rounded upper bound. A finite grid maximum alone can underestimate $L$, particularly when the density's steepest slope is near zero. Fitting a lognormal model to observed ratios does not certify the true conditional null density. The fixed-sample simulations use the analytic bound and, if necessary, raise it to $4/9$; the sequential simulations use their conditional Gaussian model, with the Lipschitz constant recomputed for each batch size. The working scales in the retrospective S\&P~500 comparison remain descriptive plug-in values.

For scale, at $\alpha=0.05$ the Gaussian likelihood-ratio models with $\mu=0.5$ and $\mu=1$ give $L\approx3.19$ and $L\approx10.71$. The corresponding relative reductions from the Markov threshold are
\begin{equation*}
\frac{1/\alpha-T_\alpha^{(L)}}{1/\alpha}
=\frac{\sqrt2(1-\alpha)^{3/2}}{3\sqrt L}
\approx24.4\% \quad\text{and}\quad 13.3\%,
\end{equation*}
respectively. Thus the practically relevant gain is largest in the moderate-$L$ regime, while the correction necessarily vanishes at rate $L^{-1/2}$ as $L$ grows.

\subsection{Smooth Beta Likelihood-Ratio Experiment}
\label{app:beta-experiment}

To check that the power gain is not specific to the piecewise-linear family, we also generate null observations from three smooth Beta models and form e-values as likelihood ratios against the indicated Beta alternatives. In each configuration, $K=10{,}000$, $p_{\mathrm{null}}=0.8$, and the results average 200 independent repetitions. For $a=3,4,5$, the likelihood ratio is $E=2X/[a(a+2)(1-X)^2]$. We bound the absolute derivative of its induced null density on 10,000 rational subintervals of $X\in[0,1]$, inflate the certified bound by $10\%$, and round upward. The resulting $L$ values are $99.149160$, $88.677529$, and $85.298118$ (rounded upward here). Each model uses an independent PCG64 stream derived from seed 42, with the same observations for both methods and levels. \Cref{tab:beta-lr} reports the nominal levels $\alpha=0.05$ and $0.10$.

\begin{table}[t]
\caption{Smooth Beta likelihood-ratio experiments: empirical FDR and mean number of rejections.}
\label{tab:beta-lr}
\centering
\small
\begin{tabular}{lccrr}
\toprule
Null vs. alternative & $\alpha$ & Method & Empirical FDR & Rejections \\
\midrule
$\mathrm{Beta}(3,3)$ vs. $\mathrm{Beta}(4,1)$ & {0.05} & e-BH & {0.000516} & {42.16} \\
 & {0.05} & e-LC & {0.000484} & {44.03} \\
 & {0.10} & e-BH & {0.000961} & {81.35} \\
 & {0.10} & e-LC & {0.001113} & {85.76} \\
$\mathrm{Beta}(4,3)$ vs. $\mathrm{Beta}(5,1)$ & {0.05} & e-BH & {0.000451} & {40.75} \\
 & {0.05} & e-LC & {0.000738} & {42.95} \\
 & {0.10} & e-BH & {0.001234} & {81.22} \\
 & {0.10} & e-LC & {0.001383} & {85.25} \\
$\mathrm{Beta}(5,3)$ vs. $\mathrm{Beta}(6,1)$ & {0.05} & e-BH & {0.000187} & {40.24} \\
 & {0.05} & e-LC & {0.000185} & {42.44} \\
 & {0.10} & e-BH & {0.001063} & {77.94} \\
 & {0.10} & e-LC & {0.001132} & {82.18} \\
\bottomrule
\end{tabular}
\end{table}

Across all six comparisons, e-LC produces more discoveries while empirical FDR remains well below the nominal level. This experiment is deliberately separated from the piecewise-linear benchmark and the extremal construction used in the sharpness proof.

\subsection{Empirical FDR Validation for Multiple Testing}
\label{app:fdr-validation-multiple}
{To complement the power comparison in \Cref{subsec:exp_multiple}, we evaluate e-LC using Gaussian likelihood-ratio e-values $E=\exp(4X-8)$. Each run uses $K=10{,}000$ hypotheses: $8{,}000$ null observations from $\mathcal N(0,1)$ and $2{,}000$ alternatives from $\mathcal N(4,1)$. We evaluate 100 equally spaced levels $\alpha\in[0.02,0.10]$ on the same 20 independent repetitions (MT19937 seed 42). Empirical FDR is the mean FDP, and the shaded band is one population standard deviation across repetitions. The analytic Lipschitz constant is $L\approx7.20754\times10^{19}$, evaluated with outward-rounded arithmetic. The resulting correction is negligible: e-LC and e-BH have identical rejection sets in these repetitions.}
In this setup, e-values are generated independently across hypotheses, so this validation is aligned with the independence regime of e-LC in \Cref{prop:lipschitz_fdr}.

{As shown in \Cref{fig:empirical-fdr-multiple}, the empirical FDR curve stays below the diagonal reference $y=\alpha$ across the tested range, so no empirical FDR inflation is observed at the tested levels. This finite simulation is a diagnostic; validity follows from the theorem only when its assumptions hold.}

\subsection{Sensitivity to Underestimating the Lipschitz Constant}
\label{app:l-misspecification}

We next study misspecification in Gaussian likelihood-ratio models. Under $H_0$, $X\sim\mathcal N(0,1)$ and $E=\exp(\mu X-\mu^2/2)$, giving $L_{\mathrm{true}}\approx3.19$ for $\mu=0.5$ and $L_{\mathrm{true}}\approx10.71$ for $\mu=1$. Each run contains $K=100{,}000$ hypotheses with null proportion $0.8$, and empirical FDR is averaged over 200 repetitions. We deliberately use
\begin{equation*}
L_{\mathrm{used}}=L_{\mathrm{true}}/c
\end{equation*}
{and scan 501 geometrically spaced factors $c\in[10,105]$. For each nominal level, \Cref{tab:l-misspecification} reports the largest grid value $c$ for which empirical FDR does not exceed $\alpha$ at that value and at all smaller factors on the tested grid.}

\begin{table}[t]
\caption{Largest empirically non-inflating underestimation factor $c=L_{\mathrm{true}}/L_{\mathrm{used}}$ on the tested grid. {Original simulation code and the alternative distribution remain to be verified.}}
\label{tab:l-misspecification}
\centering
\small
\begin{tabular}{cccc}
\toprule
$\mu\backslash\alpha$ & $0.01$ & $0.10$ & $0.20$ \\
\midrule
$0.5$ & $14.3\times$ & $14.2\times$ & $14.1\times$ \\
$1.0$ & $47.9\times$ & $46.5\times$ & $45.0\times$ \\
\bottomrule
\end{tabular}
\end{table}

No empirical FDR inflation is observed up to the factors in \Cref{tab:l-misspecification}. This is a finite-grid diagnostic, not a robustness theorem: underestimating $L$ violates the sufficient condition used in our proofs, and the transition beyond the reported grid boundary is abrupt in this experiment. The table therefore quantifies behavior in these two simulated models but does not justify replacing a conservative upper bound by an empirical underestimate.

\begin{figure}[h]
    \centering
    \subfloat[Empirical FDR]{%
\begin{tikzpicture}[baseline=(current axis.south)]\begin{axis}[evalue axis,scale only axis,width=.20\textwidth,height=2.25cm,xticklabel style={text height=1.5ex,text depth=.3ex},xlabel style={text height=1.5ex,text depth=.3ex},xlabel={$\alpha$},ylabel={Empirical FDR},xmin=.02,xmax=.1,enlarge x limits=false,ymode=log,log basis y=10,xtick={.02,.06,.1},legend pos=south east]
\addplot[name path=empiricalfdru,draw=none,forget plot] coordinates {(0.02,0.002242) (0.020808,0.002223) (0.021616,0.002232) (0.022424,0.002398) (0.023232,0.002447) (0.02404,0.002433) (0.024848,0.002514) (0.025657,0.002629) (0.026465,0.002615) (0.027273,0.002793) (0.028081,0.002864) (0.028889,0.002853) (0.029697,0.002844) (0.030505,0.002863) (0.031313,0.002854) (0.032121,0.002888) (0.032929,0.002929) (0.033737,0.002922) (0.034545,0.002981) (0.035354,0.002973) (0.036162,0.003051) (0.03697,0.003043) (0.037778,0.003308) (0.038586,0.003525) (0.039394,0.003495) (0.040202,0.003488) (0.04101,0.003579) (0.041818,0.00358) (0.042626,0.00362) (0.043434,0.003637) (0.044242,0.003629) (0.045051,0.003671) (0.045859,0.003724) (0.046667,0.00379) (0.047475,0.00382) (0.048283,0.003995) (0.049091,0.003988) (0.049899,0.004062) (0.050707,0.004131) (0.051515,0.004171) (0.052323,0.004164) (0.053131,0.004282) (0.053939,0.004385) (0.054747,0.004424) (0.055556,0.004441) (0.056364,0.004471) (0.057172,0.004465) (0.05798,0.00446) (0.058788,0.004593) (0.059596,0.00465) (0.060404,0.004643) (0.061212,0.004637) (0.06202,0.004664) (0.062828,0.004792) (0.063636,0.004787) (0.064444,0.00497) (0.065253,0.005055) (0.066061,0.005111) (0.066869,0.005169) (0.067677,0.005198) (0.068485,0.005275) (0.069293,0.005293) (0.070101,0.005367) (0.070909,0.005449) (0.071717,0.005557) (0.072525,0.005706) (0.073333,0.005697) (0.074141,0.005692) (0.074949,0.005794) (0.075758,0.005884) (0.076566,0.005879) (0.077374,0.00591) (0.078182,0.005931) (0.07899,0.005972) (0.079798,0.006021) (0.080606,0.006018) (0.081414,0.006089) (0.082222,0.006209) (0.08303,0.006256) (0.083838,0.006251) (0.084646,0.006248) (0.085455,0.006387) (0.086263,0.00639) (0.087071,0.0064) (0.087879,0.006475) (0.088687,0.006495) (0.089495,0.006491) (0.090303,0.006587) (0.091111,0.006583) (0.091919,0.006579) (0.092727,0.006645) (0.093535,0.006694) (0.094343,0.006687) (0.095152,0.006728) (0.09596,0.006794) (0.096768,0.006789) (0.097576,0.006849) (0.098384,0.006891) (0.099192,0.006888) (0.1,0.006882)};
\addplot[name path=empiricalfdrl,draw=none,forget plot] coordinates {(0.02,0.000452) (0.020808,0.000527) (0.021616,0.000574) (0.022424,0.000608) (0.023232,0.000685) (0.02404,0.000683) (0.024848,0.000656) (0.025657,0.000665) (0.026465,0.000663) (0.027273,0.000747) (0.028081,0.000796) (0.028889,0.000793) (0.029697,0.00079) (0.030505,0.000827) (0.031313,0.000824) (0.032121,0.00084) (0.032929,0.000919) (0.033737,0.000918) (0.034545,0.000913) (0.035354,0.000911) (0.036162,0.000889) (0.03697,0.000885) (0.037778,0.000874) (0.038586,0.000971) (0.039394,0.001119) (0.040202,0.001116) (0.04101,0.001273) (0.041818,0.001326) (0.042626,0.001342) (0.043434,0.001377) (0.044242,0.001375) (0.045051,0.001453) (0.045859,0.001516) (0.046667,0.001504) (0.047475,0.001526) (0.048283,0.001599) (0.049091,0.001596) (0.049899,0.001638) (0.050707,0.001687) (0.051515,0.001699) (0.052323,0.001696) (0.053131,0.001818) (0.053939,0.001829) (0.054747,0.001906) (0.055556,0.001943) (0.056364,0.001967) (0.057172,0.001963) (0.05798,0.00196) (0.058788,0.002009) (0.059596,0.002064) (0.060404,0.002061) (0.061212,0.002119) (0.06202,0.002208) (0.062828,0.002194) (0.063636,0.002193) (0.064444,0.002182) (0.065253,0.002211) (0.066061,0.002269) (0.066869,0.002327) (0.067677,0.002352) (0.068485,0.002389) (0.069293,0.002425) (0.070101,0.002467) (0.070909,0.002499) (0.071717,0.002507) (0.072525,0.00253) (0.073333,0.002527) (0.074141,0.002524) (0.074949,0.002536) (0.075758,0.002558) (0.076566,0.002555) (0.077374,0.002578) (0.078182,0.002669) (0.07899,0.002682) (0.079798,0.002745) (0.080606,0.002744) (0.081414,0.002723) (0.082222,0.002777) (0.08303,0.002782) (0.083838,0.002841) (0.084646,0.00284) (0.085455,0.002877) (0.086263,0.002928) (0.087071,0.00309) (0.087879,0.003069) (0.088687,0.003101) (0.089495,0.003099) (0.090303,0.003115) (0.091111,0.003111) (0.091919,0.003109) (0.092727,0.003151) (0.093535,0.003156) (0.094343,0.003153) (0.095152,0.003226) (0.09596,0.003272) (0.096768,0.003269) (0.097576,0.003257) (0.098384,0.003327) (0.099192,0.003324) (0.1,0.003322)};
\addplot[plotblue,fill opacity=.13,forget plot] fill between[of=empiricalfdru and empiricalfdrl];
\addplot[plotblue,solid,thick,mark=triangle*,mark size=1.3pt,mark repeat=20,mark options={solid,fill=white}] coordinates {(0.02,0.001347) (0.020808,0.001375) (0.021616,0.001403) (0.022424,0.001503) (0.023232,0.001566) (0.02404,0.001558) (0.024848,0.001585) (0.025657,0.001647) (0.026465,0.001639) (0.027273,0.00177) (0.028081,0.00183) (0.028889,0.001823) (0.029697,0.001817) (0.030505,0.001845) (0.031313,0.001839) (0.032121,0.001864) (0.032929,0.001924) (0.033737,0.00192) (0.034545,0.001947) (0.035354,0.001942) (0.036162,0.00197) (0.03697,0.001964) (0.037778,0.002091) (0.038586,0.002248) (0.039394,0.002307) (0.040202,0.002302) (0.04101,0.002426) (0.041818,0.002453) (0.042626,0.002481) (0.043434,0.002507) (0.044242,0.002502) (0.045051,0.002562) (0.045859,0.00262) (0.046667,0.002647) (0.047475,0.002673) (0.048283,0.002797) (0.049091,0.002792) (0.049899,0.00285) (0.050707,0.002909) (0.051515,0.002935) (0.052323,0.00293) (0.053131,0.00305) (0.053939,0.003107) (0.054747,0.003165) (0.055556,0.003192) (0.056364,0.003219) (0.057172,0.003214) (0.05798,0.00321) (0.058788,0.003301) (0.059596,0.003357) (0.060404,0.003352) (0.061212,0.003378) (0.06202,0.003436) (0.062828,0.003493) (0.063636,0.00349) (0.064444,0.003576) (0.065253,0.003633) (0.066061,0.00369) (0.066869,0.003748) (0.067677,0.003775) (0.068485,0.003832) (0.069293,0.003859) (0.070101,0.003917) (0.070909,0.003974) (0.071717,0.004032) (0.072525,0.004118) (0.073333,0.004112) (0.074141,0.004108) (0.074949,0.004165) (0.075758,0.004221) (0.076566,0.004217) (0.077374,0.004244) (0.078182,0.0043) (0.07899,0.004327) (0.079798,0.004383) (0.080606,0.004381) (0.081414,0.004406) (0.082222,0.004493) (0.08303,0.004519) (0.083838,0.004546) (0.084646,0.004544) (0.085455,0.004632) (0.086263,0.004659) (0.087071,0.004745) (0.087879,0.004772) (0.088687,0.004798) (0.089495,0.004795) (0.090303,0.004851) (0.091111,0.004847) (0.091919,0.004844) (0.092727,0.004898) (0.093535,0.004925) (0.094343,0.00492) (0.095152,0.004977) (0.09596,0.005033) (0.096768,0.005029) (0.097576,0.005053) (0.098384,0.005109) (0.099192,0.005106) (0.1,0.005102)};
\addlegendentry{e-LC}
\addplot[plotred,solid,thick] coordinates {(0.02,0.02) (0.020808,0.020808) (0.021616,0.021616) (0.022424,0.022424) (0.023232,0.023232) (0.02404,0.02404) (0.024848,0.024848) (0.025657,0.025657) (0.026465,0.026465) (0.027273,0.027273) (0.028081,0.028081) (0.028889,0.028889) (0.029697,0.029697) (0.030505,0.030505) (0.031313,0.031313) (0.032121,0.032121) (0.032929,0.032929) (0.033737,0.033737) (0.034545,0.034545) (0.035354,0.035354) (0.036162,0.036162) (0.03697,0.03697) (0.037778,0.037778) (0.038586,0.038586) (0.039394,0.039394) (0.040202,0.040202) (0.04101,0.04101) (0.041818,0.041818) (0.042626,0.042626) (0.043434,0.043434) (0.044242,0.044242) (0.045051,0.045051) (0.045859,0.045859) (0.046667,0.046667) (0.047475,0.047475) (0.048283,0.048283) (0.049091,0.049091) (0.049899,0.049899) (0.050707,0.050707) (0.051515,0.051515) (0.052323,0.052323) (0.053131,0.053131) (0.053939,0.053939) (0.054747,0.054747) (0.055556,0.055556) (0.056364,0.056364) (0.057172,0.057172) (0.05798,0.05798) (0.058788,0.058788) (0.059596,0.059596) (0.060404,0.060404) (0.061212,0.061212) (0.06202,0.06202) (0.062828,0.062828) (0.063636,0.063636) (0.064444,0.064444) (0.065253,0.065253) (0.066061,0.066061) (0.066869,0.066869) (0.067677,0.067677) (0.068485,0.068485) (0.069293,0.069293) (0.070101,0.070101) (0.070909,0.070909) (0.071717,0.071717) (0.072525,0.072525) (0.073333,0.073333) (0.074141,0.074141) (0.074949,0.074949) (0.075758,0.075758) (0.076566,0.076566) (0.077374,0.077374) (0.078182,0.078182) (0.07899,0.07899) (0.079798,0.079798) (0.080606,0.080606) (0.081414,0.081414) (0.082222,0.082222) (0.08303,0.08303) (0.083838,0.083838) (0.084646,0.084646) (0.085455,0.085455) (0.086263,0.086263) (0.087071,0.087071) (0.087879,0.087879) (0.088687,0.088687) (0.089495,0.089495) (0.090303,0.090303) (0.091111,0.091111) (0.091919,0.091919) (0.092727,0.092727) (0.093535,0.093535) (0.094343,0.094343) (0.095152,0.095152) (0.09596,0.09596) (0.096768,0.096768) (0.097576,0.097576) (0.098384,0.098384) (0.099192,0.099192) (0.1,0.1)};
\addlegendentry{$y=\alpha$}
\end{axis}\end{tikzpicture}%
        \label{fig:empirical-fdr-multiple}%
    }
    \hfill
    \subfloat[Effect of batch size]{%
\begin{tikzpicture}[baseline=(current axis.south)]\begin{axis}[evalue axis,scale only axis,width=.20\textwidth,height=2.25cm,xticklabel style={text height=1.5ex,text depth=.3ex},xlabel style={text height=1.5ex,text depth=.3ex},xlabel={Batch size},xmin=0,xmax=1000,enlarge x limits=false,ymax=5,xtick={0,500,1000},ylabel={Mean of $M_{\max}$},legend pos=north west]
\addplot[plotblue,solid,thick,mark=triangle*,mark size=1.3pt,mark repeat=4,mark options={solid,fill=white}] coordinates {(1,1.01660517748) (2,1.02229248651) (3,1.03748816828) (5,1.05882533581) (8,1.08550518511) (12,1.11862147295) (20,1.15178559655) (30,1.17864622625) (45,1.20875288323) (60,1.26004070862) (80,1.34028912288) (100,1.44206683568) (150,1.54697510128) (200,1.64936692659) (300,1.81114215357) (500,2.05462604563) (800,2.2406632536) (1000,2.29360025222)};
\addlegendentry{$\alpha=0.1$}
\addplot[plotteal,solid,thick,mark=square*,mark size=1.3pt,mark repeat=4,mark options={solid,fill=white}] coordinates {(1,1.01023267585) (2,1.01766437986) (3,1.03478727877) (5,1.05219699105) (8,1.0687596612) (12,1.09646117409) (20,1.15283979841) (30,1.23609689867) (45,1.32094044424) (60,1.39272190894) (80,1.4473656818) (100,1.49966554865) (150,1.59943486991) (200,1.76461691921) (300,1.95099096363) (500,2.13806391656) (800,2.30624436412) (1000,2.37343943588)};
\addlegendentry{$\alpha=0.05$}
\addplot[plotred,solid,thick,mark=diamond*,mark size=1.3pt,mark repeat=4,mark options={solid,fill=white}] coordinates {(1,1.01929870166) (2,1.02660917112) (3,1.04383207015) (5,1.06539865616) (8,1.09266087952) (12,1.1347375682) (20,1.18602310293) (30,1.22166578809) (45,1.2424433693) (60,1.27119548846) (80,1.31245459932) (100,1.35382541185) (150,1.39254703198) (200,1.46854514735) (300,1.69533560276) (500,2.16888318576) (800,2.76244325209) (1000,3.0547405171)};
\addlegendentry{$\alpha=0.01$}
\end{axis}\end{tikzpicture}%
        \label{fig:batch_size_mmax}%
    }
    \hfill
    \subfloat[Sample size]{%
        \input{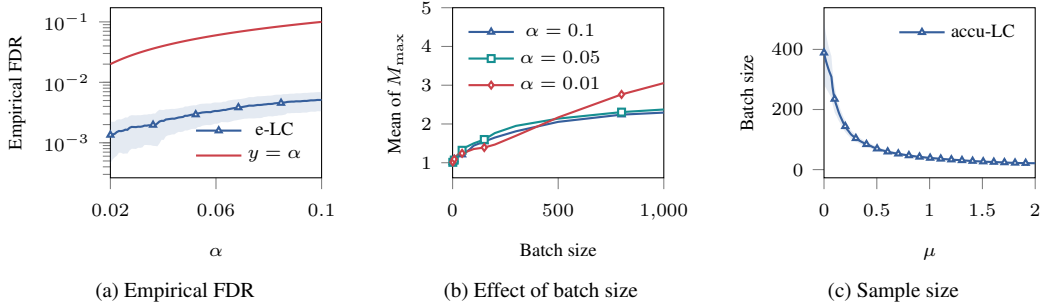}%
        \label{fig:effectiveness_plot}%
    }
    \caption{{Additional simulations. (a) Empirical FDR; shading is one standard deviation and the red line is $y=\alpha$. (b) LC null crossing probabilities by batch size, with pointwise 95\% Wilson intervals and nominal reference levels. (c) Restricted stopping samples $\min(T,H)$ for LC and EP on a logarithmic y-axis, $H=50{,}000$; shading is one Monte Carlo standard error.}}
    \label{fig:additional-experiments}
\end{figure}

\subsection{{Simulation of Sequential Inference}}\label{app_exp}
{We evaluate the process in \Cref{sec:applications} on independent Gaussian observations. The likelihood-ratio design parameter is fixed at $\theta=0.1$, distinct from the data-generating mean $\mu$. For a fresh, disjoint block of $b$ observations, the input is}
\begin{equation*}
 {E=\exp\!\left(\theta\sum_{i=1}^b X_i-b\theta^2/2\right).}
\end{equation*}
{Under the null, $X_i\sim\mathcal N(0,1)$, so $\log E\sim\mathcal N(-b\theta^2/2,b\theta^2)$. We recompute the analytic Lipschitz constant for each $b$, use an upward numerical guard of $10^{-12}$ in relative terms, and set $\lambda=0.5$, $\tau=0.5/\sqrt{2L}$, and $\eta=(1-\tau+L\tau^3/3)^{-1}$. The numerical implementation also checks the exact lognormal expectation of each update against one. Both EP and LC multiply fresh block factors along one continuing path, starting from $M_0=1$.}

\paragraph{{Null Crossing Probabilities.}}
{For each of the 18 batch sizes $b$ from 1 to 1000 shown in \Cref{fig:batch_size_mmax}, we simulate $5{,}000$ paths with a common budget of $36{,}000$ raw observations, giving $36{,}000/b$ updates. We estimate $\Pr(\max_{0\le t\le36{,}000/b}M_t\ge1/\alpha)$ at $\alpha\in\{0.1,0.05,0.01\}$. In \Cref{fig:batch_size_mmax}, LC crossing estimates stay below their nominal levels; bands are pointwise 95\% Wilson intervals across independent paths, not simultaneous confidence bands. This finite-horizon check complements the conditional-model argument, which supplies the anytime-valid guarantee. The expectation of a running maximum need not be at most one, even for a valid nonnegative supermartingale, so it is not used as a validity criterion.}

\paragraph{{Stopping Samples and Censoring.}}
{For $b=1$, generate $X_t\sim\mathcal N(\mu,1)$ at 200 equally spaced means $\mu\in[0,2]$, retaining $\theta=0.1$. Each mean uses $1{,}000$ paired EP/LC paths, $\alpha=0.1$, and horizon $H=50{,}000$ observations. We record first crossings and treat paths without a crossing by $H$ as censored; each contributes $H$ to $\min(T,H)$. \Cref{fig:effectiveness_plot} reports the mean of $\min(T,H)$ over all paths, with one Monte Carlo standard error; it is neither an uncensored expected stopping time nor an average over successful paths only. The simulation uses seed 20260926 with separate random-number streams across design settings; within each setting, EP and LC are evaluated on the same simulated paths.}

{At $\mu=0$, the LR inputs are nondegenerate because $\theta$ remains $0.1$: EP and LC cross on $99/1000$ and $17/1000$ paths, with restricted means $45{,}229.0$ and $49{,}165.6$. At $\mu=2$, both cross on all paths, with means $22.86$ and $21.73$. LC has a smaller restricted mean at 152 of the 200 tested means, but is slower at the remaining 48. Thus truncation can delay weak-signal detection; the calibration does not imply uniform sample-efficiency dominance. Per-path results and crossing probabilities are retained with the experiment code.}

\end{document}